\documentclass[12pt]{iopart}
\usepackage{graphicx}

\usepackage{iopams} 
\usepackage{amssymb}

\usepackage{amsthm}
\usepackage{bm}
\usepackage{graphicx}
\usepackage{ascmac}
\usepackage{color}
\usepackage{here}
\usepackage{ulem}

\begin{document}

\renewcommand{\theequation}{\thesection.\arabic{equation}}

\newtheorem{thm}{Theorem}[section]
\newtheorem{prp}[thm]{Proposition}
\newtheorem{lmm}[thm]{Lemma}
\newtheorem{cor}[thm]{Corollary}
\newtheorem{dfn}[thm]{Definition}
\newtheorem{rmk}[thm]{Remark}
\newtheorem{con}[thm]{Conjecture}
\newtheorem{Assum}[thm]{Assumption}
\newtheorem{mth}{Main Theorem}

\newcommand{\vc}[1]{\mbox{\boldmath $#1$}}
\newcommand{\fracd}[2]{\frac{\displaystyle #1}
{\displaystyle #2}}

\def\ni{\noindent}
\def\nn{\nonumber}
\def\bH{\begin{Huge}}
\def\eH{\end{Huge}}
\def\bL{\begin{Large}}
\def\eL{\end{Large}}
\def\bl{\begin{large}}
\def\el{\end{large}}
\def\beq{\begin{eqnarray}}
\def\eeq{\end{eqnarray}}

\def\p{\partial}
\def\k{\kappa}

\def\C{{\cal C}}
\def\T{{\cal T}}
\def\H{{\cal H}}
\def\M{{\cal M}}
\def\L{{\cal L}}
\def\Q{{\cal Q}}
\def\S{{\cal S}}
\def\P{{\bf P}}
\def\im{{\rm Im}}
\def\re{{\rm Re}}

\def\i{{\rm i}}

\def\He{H\'enon }
\def\St{Stokes curves }

\def\bn{\bigskip\noindent}
\def\pbn{\par\bigskip\noindent}
\def\n{\noindent}
\def\b{\bigskip}
\def\non{\nonumber}

\def\ve{\boldsymbol}

\newcommand{\bra}[1]{\langle{}#1{}|}
\newcommand{\eq}[1]{(\ref{#1})}

\title[Uniform hyperbolicity of a class of scattering maps]{Uniform hyperbolicity of a class of scattering maps}

\author{Hajime Yoshino$^{1}$, Normann Mertig$^{1,2}$ ~and Akira Shudo$^{1}$\footnote{Author to whom any correspondence should be addressed.}}
\address{$^{1}$ Department of Physics, Tokyo Metropolitan University, 
Minami-Osawa, Hachioji, Tokyo 192-0397, Japan}

\address{$^{2}$ Research $\&$ Development Group, Center for Exploratory Research, Hokkaido University Laboratory, Hitachi Ltd., Hokkaido, Japan}

\eads{yajime22@gmail.com, normann.mertig@hitachi-eu.com and shudo@tmu.ac.jp}

\begin{abstract}
In recent years fractal Weyl laws and related quantum eigenfunction hypothesis have been studied in a plethora of numerical model systems, called quantum maps. In some models studied there one can easily prove uniform hyperbolicity. Yet, a numerically sound method for computing quantum resonance states, did not exist.
To address this challenge, we recently introduced a new class of quantum maps \cite{NS18}. For these quantum maps, we showed that, quantum resonance states can numerically be computed using theoretically grounded methods such as complex scaling or weak absorbing potentials \cite{NS18}. However, proving uniform hyperbolicty for this class of quantum maps was not straight forward.
Going beyond that work this article generalizes the class of scattering maps and provides mathematical proofs for their uniform hyperbolicity. In particular, we show that the suggested class of two-dimensional symplectic scattering maps satisfies the topological horseshoe condition and uniform hyperbolicity. In order to prove these properties, we follow the procedure developed in the paper by Devaney and Nitecki \cite{DN79}. Specifically, uniform hyperbolicity is shown by identifying a proper region in which the non-wandering set satisfies a sufficient condition to have the so-called sector bundle or cone field. 
Since no quantum map is known where both a proof of uniform hyperbolicity and a methodologically sound method for numerically computing quantum resonance states exist simultaneously, the present result should be valuable to further test fractal Weyl laws and related topics such as chaotic eigenfunction hypothesis.

\end{abstract}

\maketitle

\section{Introduction}
\label{sec:introduction}

Two-dimensional symplectic maps are broadly used to study the nature of chaos in Hamiltonian systems. 
Not only are they good toy models for the investigation of continuous Hamiltonian dynamics but also worth exploring in their own rights. 
As dynamical systems, the systems with uniform hyperbolicity form an open set
because of the structural stability against small perturbations, 
and are well understood as compared to systems not satisfying uniform hyperbolicity \cite{Katok}.  
Several two-dimensional symplectic maps which are uniformly hyperbolic are nowadays at hand 
and they serve as ideal systems to study various aspects of chaos. 
The simplest example would be the baker's transformation and a class of automorphisms defined 
on the 2-torus, the so-called Arnold cat map. 
Even if the cat map is perturbed, the system is shown to keep uniform hyperbolicity as 
expected \cite{Ozorio}.  

In addition to the maps defined on a compact phase space, there exist also uniformly hyperbolic maps 
on noncompact phase spaces such as the two dimensional plane.
The most extensively studied system in such a category is the H\'enon map, which is classified 
as a unique polynomial 2-dimensional diffeomorphism that generates nontrivial dynamics \cite{FM}. 
A sufficient condition for the uniform hyperbolicity was first found in \cite{DN79}, and 
it had been believed that the system stays uniformly hyperbolic until the first tangency of stable and unstable manifolds. 
A technique based on complex dynamics has been applied to solve this problem \cite{BS04}, 
and a computational method provides an alternative proof that this is indeed the case \cite{Arai08}. 
Although the map is defined on the plane, the nonwandering set of the system is confined 
in a finite region, and it could physically be viewed as a model of the scattering process, in 
which chaos is generated in a scattering region. 

Two-dimensional symplectic maps have also been employed to examine 
the quantum manifestation of classical chaos \cite{Gutz,Haake,Stockmann}.  
In this context, the map with uniform hyperbolicity again plays a  paradigmatic role since 
a semiclassical trace formula exists for such systems, which allows to relate 
quantum eigenvalues with periodic orbits in 
the corresponding classical map. 
Historically, the study based on the trace formula has first focused 
on the system with a compact phase space, thereby the system is bounded 
with a discrete spectrum.  

In contrast, the semiclassical study of scattering systems has gotten a late start although 
a trace formula exists in such situations as well \cite{Gaspard,Wirzba}. 
It is interesting to note that this direction of research has been further driven 
by mathematicians. 
In particular, the {\it fractal Weyl law} has been conjectured under 
rigorous mathematical settings \cite{Sjostrand,Zworski99,Zworski04,Nonnen07a,
Zworski07,Nonnen07b,Datchev13,Nonnen14}. 
In accordance with these works, numerical 
\cite{Lin02,Zworski02,Lu03,Schomerus,Nonnen05,Ketz13,Borthwick14}, and even experimental works 
\cite{Sridhar99,Zworski12,Zworski13}, 
which have attempted to verify the conjecture, followed. 
The reader can access the subject in recent review articles 
\cite{Nonnen11,Zworski17,Novaes13} as well. 

Here the conventional Weyl law concerns the crudest quantum-to-classical correspondence claiming that the mean level density of discrete spectra is expressed asymptotically in terms of the area of the corresponding classical phase space. The fractal Weyl law is a counterpart of the Weyl law in scattering systems and states that the number of long-lived quantum resonance states should grow as a power law whose exponent is related to the Hausdorff dimension of the classical repeller. 

Since the fractal Weyl law is primarily expected to hold in uniformly hyperbolic systems, 
the actual verification to check its validity should also be performed employing 
systems with uniform hyperbolicity 
\cite{Lin02,Zworski02,Lu03,Schomerus,Nonnen05,Borthwick14}. 
In Refs. \cite{Lin02,Zworski02} the authors studied 
the system with Gaussian shaped repulsive potential and in 
Ref. \cite{Lu03} the 3-disk system has been employed, they obtained 
the results supporting the conjecture. The author in Ref. \cite{Borthwick14,Ketz22,Barkhofen22, Barkhofen23}
also numerically examined the conjecture by taking the geodesic flow 
on convex co-compact hyperbolic manifolds. 

On the other hand, although the symplectic or discrete maps constitute an important class of dynamical systems which can realize uniformly hyperbolic situations, reasonable models suitable for investigating quantum resonances have not been proposed so far. There nevertheless exists a couple of works, which were intended to confirm the fractal Weyl law numerically using open discrete maps \cite{Schomerus,Nonnen05,Ketz13}, which are known to be uniformly hyperbolic when they are closed.
This however invokes a serious issue when performing a test for the Weyl-conjecture, 
because undesired effects necessarily enter due to projective openings or the introduction of 
strong absorbing potentials. These usually cause diffraction or nonclassical effects, 
which are all beyond the treatment of the semiclassical argument, on which 
the original Weyl-conjecture is essentially based. 

In \cite{NS18}, the authors discuss the numerical computation of the quantum resonance spectra of a kicked scattering map system with analytic potential, which is free from diffraction. Since quantum resonance states are not square integrable, they cannot be expanded in a basis. In the method called {\it complex scaling}, a system on a complex contour obtained by rotating the real axis is considered. In this rotated system, some part of the quantum resonance states become square integrable and can be expanded in terms of a basis. The resulting spectrum coincides with the quantum resonance spectrum of the original system. The spectra obtained this way also coincide with the spectrum obtained by a method using a sufficiently small absorbing potential. On the other hand, the spectrum obtained using conventional methods such as projection or strong absorbing cannot separate the quantum resonance eigenvalues from the continuous spectrum. This raises a doubt that the spectra obtained from closed systems with projective openings are at all suited for the purpose of testing the fractal Weyl law. Thus, complex scaling method is more suitable.

The motivation of the present study is therefore to establish an analytic scattering map 
which satisfies the following conditions: i) uniform hyperbolicity is mathematically proven in a certain parameter limit, and ii) an effective method can be applied to rigorously compute quantum resonance spectra.
We hope it will be taken as a model not only for testing the fractal Weyl law but also for further inquiries arising in scattering problems, such as questions on the existence of a spectral gap in the quantum resonance distribution, and spatial structure of the associated metastable states in phase space (see more details in review articles \cite{Nonnen11,Zworski17,Novaes13}.) One could also obtain further insights into physical phenomena such as ionization or dissociation of atoms and molecules. 

In this context we here wish to notice that the H\'enon map has been proven to satisfy uniform hyperbolicity in certain parameter loci as mentioned above. In that, it could be regarded as a kind of scattering systems. However, its dynamics in the asymptotic region does not conform to the usual requirement of physical scattering systems. That is, far away from the scattering region the orbits of the H\'enon map do not attain free motion. Therefore, standard methods for computing quantum resonance spectra such as complex scaling, cannot be applied to the H\'enon map and render it unsuitable for investigations of quantum manifestations of uniformly hyperbolic scattering systems. First steps towards providing analytic scattering maps, which do exhibit free motion in the asymptotic region have been taken in \cite{NS18}, where a mathematically grounded framework for calculating quantum resonances of such scattering maps based on complex scaling was presented, thus illustrating and resolving issues with quantum resonance spectra computed based on projective openings. However, the uniform hyperbolicity of the scattering maps in \cite{NS18} could not be proven and had to be conjectured based on the shape of numerically computed manifolds.

In this paper, we proof uniform hyperbolicity for a class of scattering maps which generalizes \cite{NS18}. The core strategy of our proof follows exactly the strategy applied in Devaney and Nitecki \cite{DN79}: first we consider a sufficient condition for the topological horseshoe, which is just given by showing the existence of an appropriate region whose forward iterate gives rise to mutually disjointed components (a ternary topological horseshoe in the present case). Next we show that the 
the nonwandering set of the system is uniformly hyperbolic \cite{Katok,Wiggins}. 
By proving uniform hyperbolicity of the proposed scattering maps this article complements our previous work \cite{NS18} and establishes the first family of 2d scattering maps for which the following properties can be proven to hold simultaneously.
(i) Uniform hyperbolicity in a certain parameter limit.
(ii) Standard free motion in the asymptotic regions away from the scattering region.
(iii) Amenability to well-established numerical methods such as complex scaling and weak absorbing potentials as established in our previous work \cite{NS18}.
Providing a class of scattering maps with provable uniform hyperbolicity and numerically established methods for computing quantum resonance states should be valuable to further test fractal Weyl laws and related topics such as chaotic eigenfunction hypothesis.

The outline of the paper is as follows. 
In section 2, we introduce a class of scattering maps to be analyzed in the paper. We consider two 
types of potentials, given by a Gaussian and a Lorentzian function respectively. 
Both maps are characterized by three parameters, 
and we examine in which parameter regions the horseshoe for the nonwandering set and uniform hyperbolicity 
are achieved. 
In section 3, we consider the topological horseshoe for a certain region of scattering maps introduced in section 2, and provide a sufficient condition for it.  
In section 4, we show that the filtration property holds in a manner, similar to the H\'enon map,
and also prove that the nonwandering set of the map is contained in a 
region, for which we will show the uniform hyperbolicity. 
In section 5,  we prove two lemmas which lead to the uniform hyperbolicity of the map. 
As is well known, 
there is a sufficient condition for the system to be uniformly hyperbolic involving the so-called sector bundle or cone field \cite{DN79,Katok,Wiggins} and we call the sufficient condition as "the sector condition" in this paper.
In section 6,  we explore in which parameter regions the sector condition holds. 
Our strategy to prove the sector condition is first to find a proper region in phase space containing 
the nonwandering set of the system, and then to divide such a region into three subregions. 
For each region we separately check the conditions ensuring the sector condition, and derive a set of 
inequalities. 
In section 7, putting together all the conditions derived in section 6,  we give our main 
results which state sufficient conditions for the topological 
horseshoe and the uniform hyperbolicity in each potential case.  
Appendices are devoted to finding redundant inequalities and further specifying 
parameter regions in which the horseshoe and hyperbolicity are realized. 

\section{Scattering Map}
\subsection{Map}
\label
{sec:scattering_map}
\setcounter{equation}{0}

We here consider a two-dimensional blacksymplectic map:
\renewcommand{\arraystretch}{1.5}
\begin{eqnarray}
\label{eq:map1}
U: 
\left(\begin{array}{c}
q_{n+1} \\
p_{n+1} 
\end{array}\right)
=
\left(\begin{array}{c}
\displaystyle q_n + p_n - V'(q_n)\\
\displaystyle p_n - V'(q_n)
\end{array}\right). 
\end{eqnarray}
Here one can easily check that the map is symplectic, irrespective of 
the form of $V(q)$,  by 
calculating the Jacobian of the map $U$.
We hereafter take 
\begin{eqnarray}
V(q) = \kappa \left( V_{1}(q) + \varepsilon V_{2}(q) \right), 
\end{eqnarray}
\begin{eqnarray}
\label{eq:potential}
V_1(q) = - v(q^2), \\
V_{2}(q) = -
\left(
\int_{0}^{q} v((q'-q_b)^2) dq' -  \int_{0}^{q} v((q'+q_b)^2)  dq'
\right). 
\end{eqnarray}
In particular, we investigate the following two types of potentials: 
\begin{eqnarray}
\label{eq:Gaussian}
v^{(G)}(q^2) = e^{-q^2},  \\
\label{eq:Lorentzian}
v^{(L)}(q^2) = \frac{\displaystyle 1}{\displaystyle 1 + q^2}. 
\end{eqnarray}
Below we drop the superscript $(G)$ or $(L)$ if formulas or equations 
under consideration hold for both cases. 
Here note that $V(q)=V(-q)$ holds also for both choices. 
We hereafter assume 
\begin{eqnarray}
\k >0, 
\end{eqnarray}
which provides physically reasonable scattering systems. 
For $\k \le 0$, the system becomes totally different, 
and the subsequent analysis will not work. 
The parameter $\varepsilon$ is expressed 
in terms of other parameters $q_b$ and $q_f$ as 
\begin{eqnarray}
\label{eq:epsilon}
\varepsilon = -
\frac{V_1'(q_f)}{V_2'(q_f)} 
= -\frac{2q_fv'(q_f^2)}{v \bigl( (q_f-q_b)^2) - v ((q_f+q_b)^2 \bigr)  }
.
\end{eqnarray}
We assume that the parameters $q_b$ and $q_f$ are both 
real and positive. In particular, the parameter $q_f$ is set to be
 a zero of the function $V'(q)$. 
In \ref{app:Gaussian_out} and \ref{app:Lorentzian_out}, we show that 
the function $V'(q)$ has zeros only at $-q_f, 0$ and $q_f$ and the property 
\begin{eqnarray}
\label{eq:cond_outgoing}
V'(q) < 0 ~~{\rm for}~~ q > q_f,
\end{eqnarray}
holds, which also leads to $V'(q) > 0 ~~{\rm for}~~ q < -q_f$ due to the symmetry. 
For the Lorentzian case, we need the following condition to prove it (see figure B1), 
\begin{equation}
\label{eq:condout}
q_f < \sqrt{\frac{1}{2} q_b^2 + 1}. 
\end{equation}
As a result, 
the potential function $V(q)$ has a single well around $q=0$ and 
bumps around $q= \pm q_f$. 
Note that $\varepsilon$ does not depend on the parameter $\k$. 
The potential function $V(q)$ is thereby characterized by 
the three parameters $\kappa, q_b$ and $q_f$. 

The inverse map is written as 
\renewcommand{\arraystretch}{1.5}
\begin{eqnarray}
\label{eq:map1_inverse}
U^{-1}: 
\left(\begin{array}{c}
q_{n} \\
p_{n} 
\end{array}\right)
=
\left(\begin{array}{c}
\displaystyle q_{n+1}  -p_{n+1} \\
\displaystyle p_{n+1}+ V'(q_{n})  
\end{array}\right). 
\end{eqnarray}
The
determinant
 of the Jacobian for the inverse map is also unity as the forward map $U$.
We first illustrate in Fig. \ref{fig:potential} the form of the potential function in each case and 
also stable and unstable manifolds for fixed points in Fig. \ref{fig:phase_space}.
Here we have chosen a set of parameters so that 
the topological horseshoe and the uniform hyperbolicty are realized at least 
within numerical computations. 
\begin{figure}[H]

  \centering
    \begin{tabular}{c}
      \begin{minipage}{0.50\hsize}
        \centering
	\includegraphics[width = 8.0cm,bb =0 0 461 346]{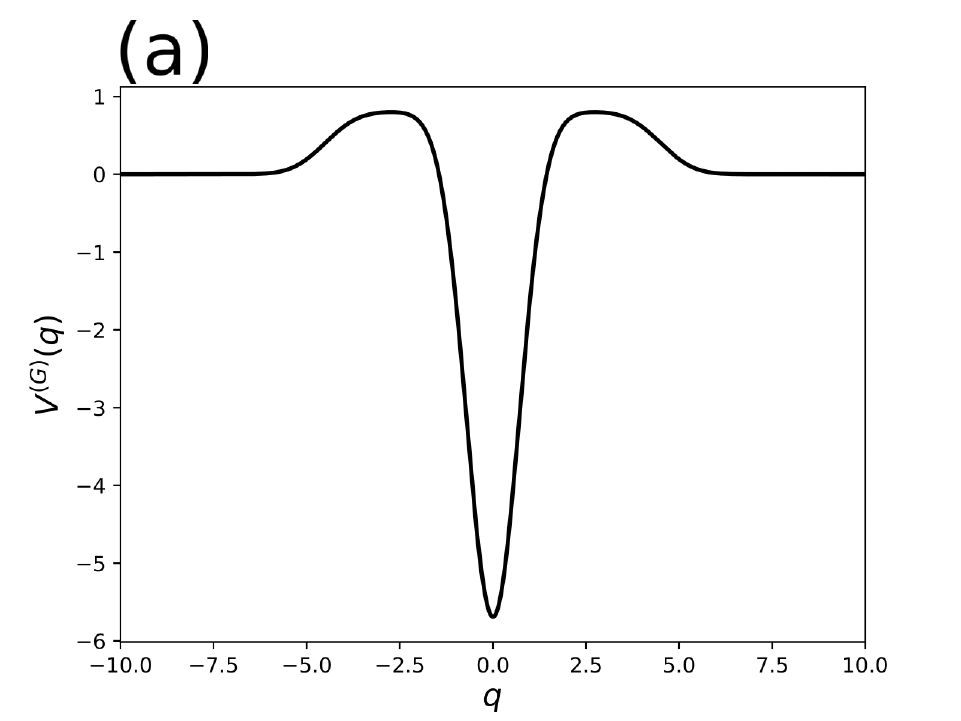}
      \end{minipage}
 
      \begin{minipage}{0.50\hsize}
        \centering
	\includegraphics[width = 8.0cm,bb = 0 0 461 346]{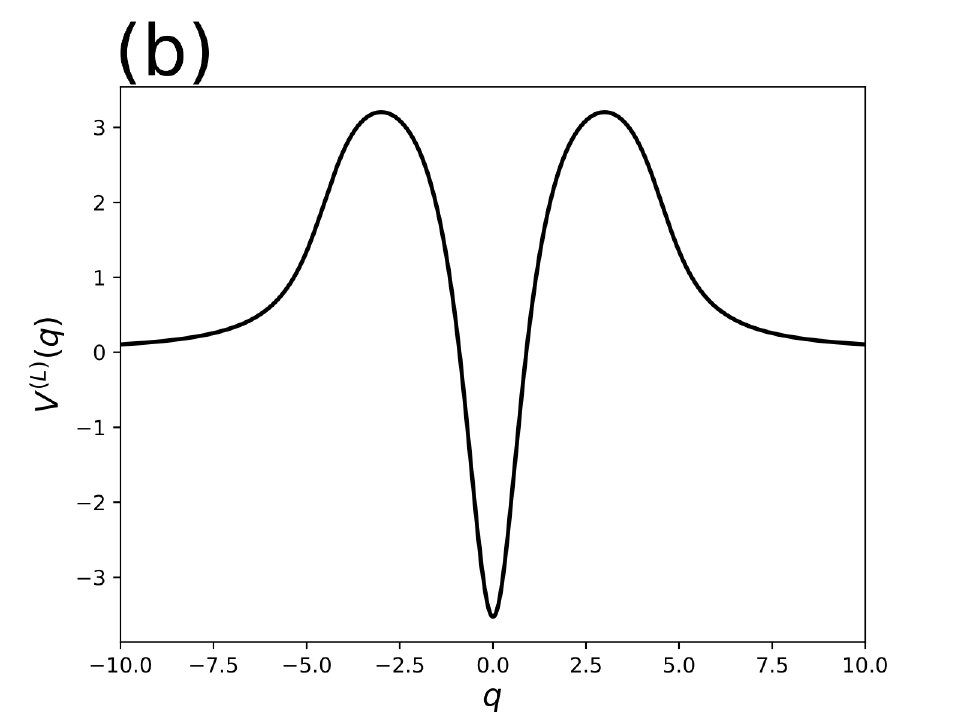}
      \end{minipage} 
    \end{tabular}
    	\caption{The form of potential function for (a) Gaussian  function case with the parameters are taken as $(q_b,q_f,\k) = (4.5,3.0,6.5)$ and (b) Lorentzian function case with the parameters are taken as $(q_b,q_f,\k) = (4.5,3.0,8.0)$. 
	The central valley and side bumps are respectively formed by the function $V_1(q)$ and $V_2(q)$. 
	\label{fig:potential}}
\end{figure} 

\begin{figure}[H]
  \centering
    \begin{tabular}{c}
      \begin{minipage}{0.50\hsize}
        \centering
	\includegraphics[width = 8.0cm,bb = 0 0 461 346]{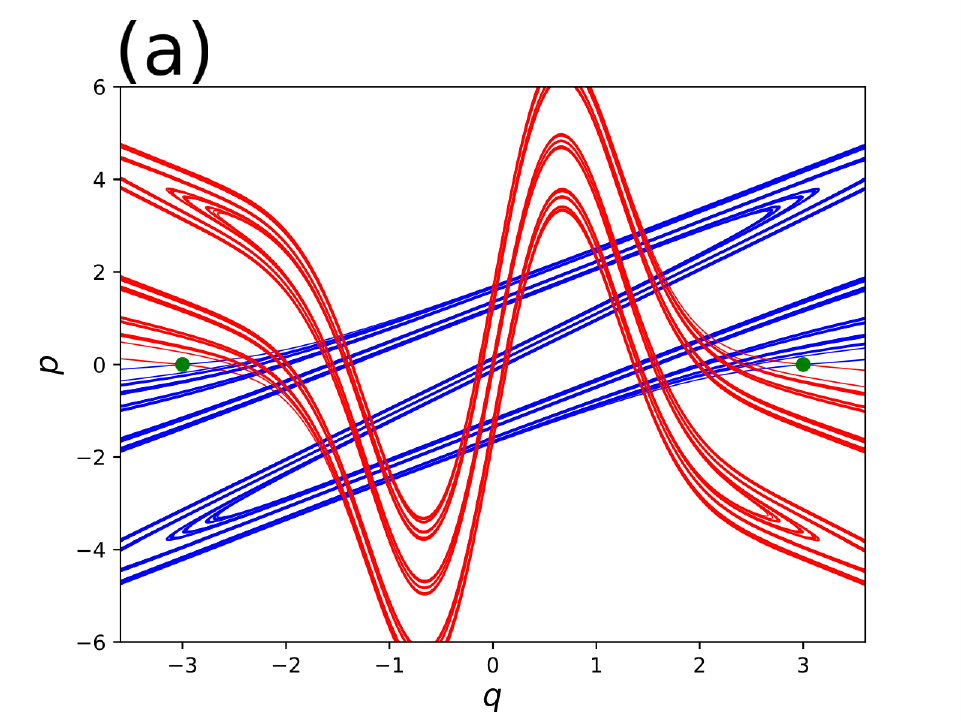}
      \end{minipage}
 
      \begin{minipage}{0.50\hsize}
        \centering
	\includegraphics[width = 8.0cm,bb = 0 0 461 346]{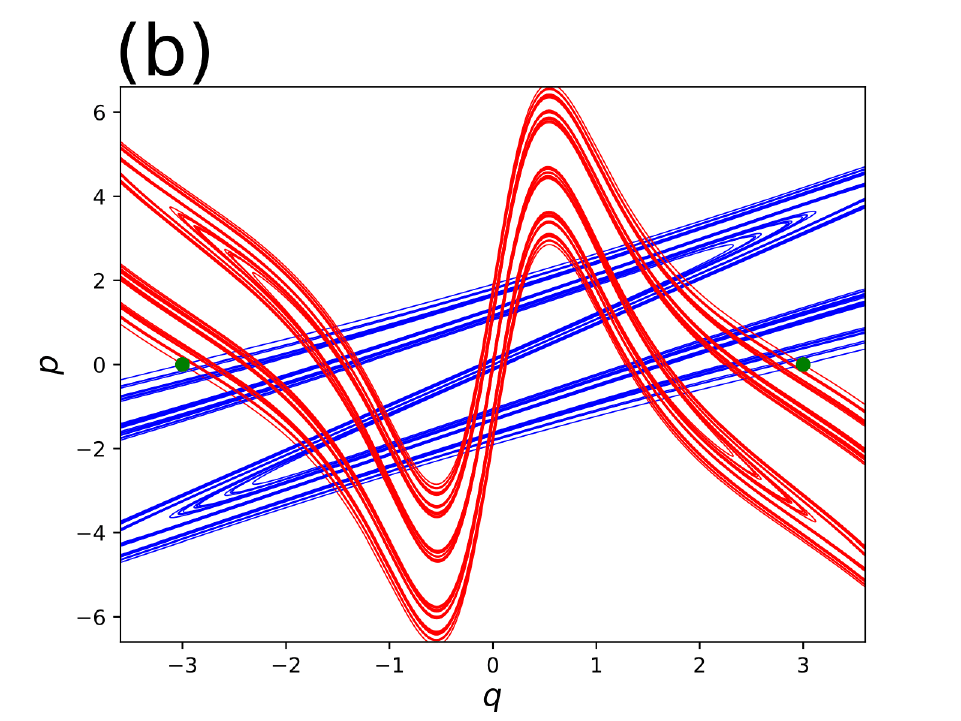}
      \end{minipage} 
    \end{tabular}
    	\caption{
	Stable (red) and unstable (blue) manifolds associated with fixed points for (a) the Gaussian $(q_b,q_f,\k) = (4.5,3.0,6.5)$ and (b) Lorentzian function  $(q_b,q_f,\k) = (4.5,3.0,8.0)$ case.  Green dots represent the fixed points of the map. 
	\label{fig:phase_space}}

\end{figure}


\subsection{Transformation for the map}
We here consider a coordinate transformation
$T:(q,p)\rightarrow (\widetilde{q},\widetilde{p})$
\begin{eqnarray}
    \left( \begin{array}{c}
       \widetilde{q} \\
      \widetilde{p}
    \end{array}
    \right)
&:=&T
  \left( 
    \begin{array}{c}
       q \\
      p
    \end{array}
    \right), 
\end{eqnarray}
where 
\begin{eqnarray}
T &:=& \left(
    \begin{array}{cc}
       1&0 \\
      -1& 1
    \end{array}
  \right) . 
\end{eqnarray}
Later on this transformation will allow for defining a scattering region which is a unit square.
According to this transformation, the original map $U$
is transformed into the map $\widetilde{U}$: 
\renewcommand{\arraystretch}{1.5}
\begin{eqnarray}
\label{eq:trans_map1}
\widetilde{U}: 
\left(\begin{array}{c}
\widetilde{q}_{n+1} \\
\widetilde{p}_{n+1} 
\end{array}\right)
=
\left(\begin{array}{c}
\displaystyle 2\widetilde{q}_{n} + \widetilde{p}_{n} - V'(\widetilde{q}_{n})   \\
\displaystyle -\widetilde{q}_{n} 
\end{array}\right), 
\end{eqnarray}
and the inverse map is given as
\renewcommand{\arraystretch}{1.5}
\begin{eqnarray}
\label{eq:trans_map1_inverse}
\widetilde{U}^{-1}: 
\left(\begin{array}{c}
\widetilde{q}_{n} \\
\widetilde{p}_{n} 
\end{array}\right)
=
\left(\begin{array}{c}
\displaystyle -\widetilde{p}_{n+1}  \\
\displaystyle \widetilde{q}_{n+1} + 2 \widetilde{p}_{n+1} + V'(\widetilde{q}_{n}) 
\end{array}\right). 
\end{eqnarray}
In the following, we rewrite $(\widetilde{q},\widetilde{p})$ as $(q,p)$, 
and $\widetilde{U}$ as $U$, respectively. 
\renewcommand{\arraystretch}{1.5}
\begin{eqnarray}
\label{eq:Map}
U: 
\left(\begin{array}{c}
q_{n+1} \\
p_{n+1} 
\end{array}\right)
=
\left(\begin{array}{c}
\displaystyle 2q_{n} + p_{n}  - V'(q_n) \\
\displaystyle -q_{n} 
\end{array}\right). 
\end{eqnarray}
The associated tangent map is given by
\begin{eqnarray}
DU_{(q_n,p_n)} &=& \left(
    \begin{array}{cc}
       \frac{\displaystyle \partial q_{n+1}}{\displaystyle \partial q_{n}}& \frac{\displaystyle \partial q_{n+1}}{\displaystyle \partial p_{n}} \\
      \frac{\displaystyle \partial p_{n+1}}{\displaystyle \partial q_{n}} & \frac{\displaystyle \partial p_{n+1}}{\displaystyle \partial p_{n}}
    \end{array}
  \right) 
\nonumber \\  &=&
\left(
    \begin{array}{cc}
       2 - V''(q_{n})& 1 \\
      -1& 0
    \end{array}
  \right) 
  \\
DU_{(q_n,p_n)}^{-1}  &=& \left(
    \begin{array}{cc}
       \frac{\displaystyle \partial q_{n-1}}{\displaystyle \partial q_{n}}& \frac{\displaystyle \partial q_{n-1}}{\displaystyle \partial p_{n}} \\
      \frac{\displaystyle \partial p_{n-1}}{\displaystyle \partial q_{n}} & \frac{\displaystyle \partial p_{n-1}}{\displaystyle \partial p_{n}}
    \end{array}
  \right) 
\nonumber \\ &=& \left(
    \begin{array}{cc}
      0 & -1 \\
      1 &  2 - V''(q_{n})
    \end{array}
  \right) . 
\end{eqnarray}

\subsection{Scattering Region}
We now illustrate the scattering region introduced in \cite{NS18}, and show how 
the horseshoe shaped region is created under the dynamics. 
As shown in Fig. \ref{fig:Scattering_region}, 
the scattering region $\mathcal{S}$ is a region whose boundaries 
are given by segments of stable and unstable manifolds (black curves in the figure). 
Figure \ref{fig:Horseshoe_1} illustrates 
that the intersection with its forward image, i.e., $\mathcal{S} \cap U(\mathcal{S})$, 
are composed of three mutually disjoint regions.
\begin{eqnarray}
\mathcal{S} \cap U(\mathcal{S})  = \mathcal{X}_1 \cup \mathcal{Y}_1 \cup \mathcal{Z}_1,  
\end{eqnarray}
where
\begin{eqnarray}
\mathcal{X}_1 \cap \mathcal{Y}_1  = \emptyset,\mathcal{Y}_1 \cap \mathcal{Z}_1  = \emptyset,\mathcal{Z}_1 \cap \mathcal{X}_1  = \emptyset.
\end{eqnarray}
In a similar manner, as shown in Fig. \ref{fig:Horseshoe_2}, 
the intersection with the backward image, i.e., 
$\mathcal{S} \cap U^{-1}(\mathcal{S})$
are composed of three mutually disjoint regions. 
\begin{eqnarray}
\mathcal{S} \cap U^{-1}(\mathcal{S})  = \mathcal{X}_0 \cup \mathcal{Y}_0 \cup \mathcal{Z}_0,  
\end{eqnarray}
where
\begin{eqnarray}
\mathcal{X}_0 \cap \mathcal{Y}_0  = \emptyset,\mathcal{Y}_0 \cap \mathcal{Z}_0  = \emptyset,\mathcal{Z}_0 \cap \mathcal{X}_0  = \emptyset . 
\end{eqnarray}
Here, as demonstrated in Figs. \ref{fig:Horseshoe_1}(a) and \ref{fig:Horseshoe_2}(b), 
the boundaries of the regions $\mathcal{X}_0,\mathcal{Y}_0$ and $\mathcal{Z}_0$ 
are composed of four portions, each of which is respectively 
given by the stable and unstable manifolds for the right-side fixed point, and those for the left-side fixed point. The regions 
$\mathcal{X}_1,\mathcal{Y}_1$ and $\mathcal{Z}_1$ are the images of 
$\mathcal{X}_0,\mathcal{Y}_0$ and $\mathcal{Z}_0$, 
i.e., 
$\mathcal{X}_1=U(\mathcal{X}_0),\mathcal{Y}_1=U(\mathcal{Y}_0)$
and $\mathcal{Z}_1=U(\mathcal{Z}_0)$.
As illustrated in Figs. \ref{fig:Horseshoe_1} (a),(b) and Figs. \ref{fig:Horseshoe_2} (a),(b), 
forward and backward images form twice folded horseshoe shape regions, not like the 
well-known once folded horseshoe, so one would 
expect that the scattering region $\mathcal{S}$ satisfies the topological horseshoe condition 
in a parameter regime where the black regions in Figs. \ref{fig:Horseshoe_1} and \ref{fig:Horseshoe_2} 
emerge.
However, since we do not have any compact expressions for the stable and unstable manifolds, 
and intersection points at hand, we have to find an analytically tractable region, 
instead of the scattering region $\mathcal{S}$,  to prove that 
the topological horseshoe is indeed achieved.

\begin{figure}[H]

  \centering
        \centering
	\includegraphics[width = 8.0cm,bb = 0 0 461 346]{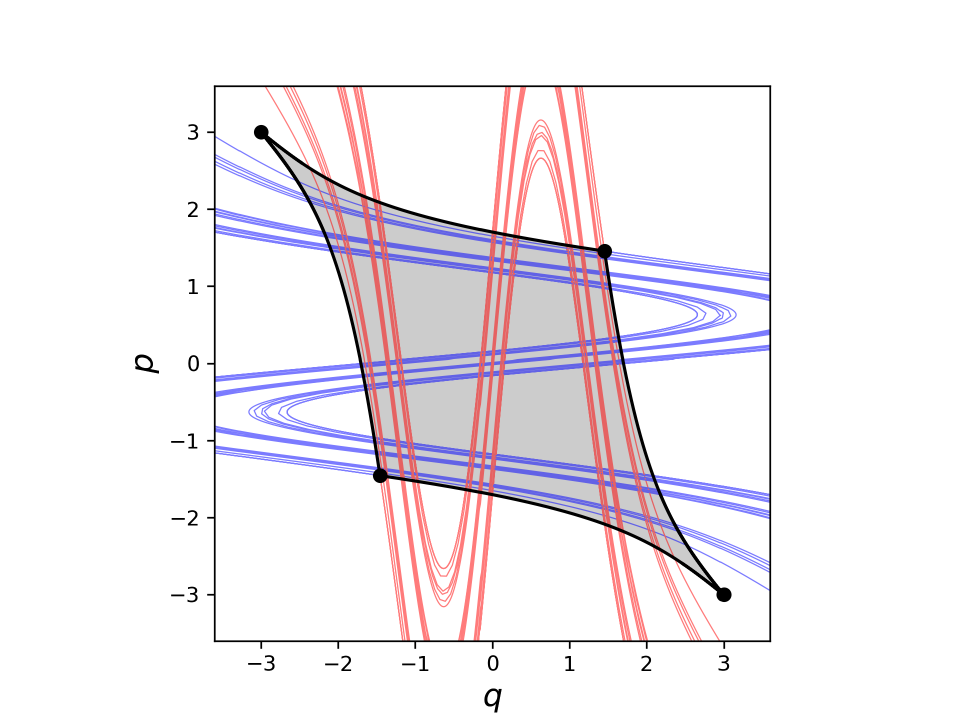}
	\caption{Illustration of the scattering region $\mathcal{S}$(gray). Boundaries (thick black curves) of the region consist of segments of the stable and unstable manifolds (red and blue respectively). Dots represent the fixed points and the heteroclinic points. }
	\label{fig:Scattering_region}
\end{figure} 
\begin{figure}[H]
  \centering
	\includegraphics[width = 15.0cm,bb =0 0 850 346]{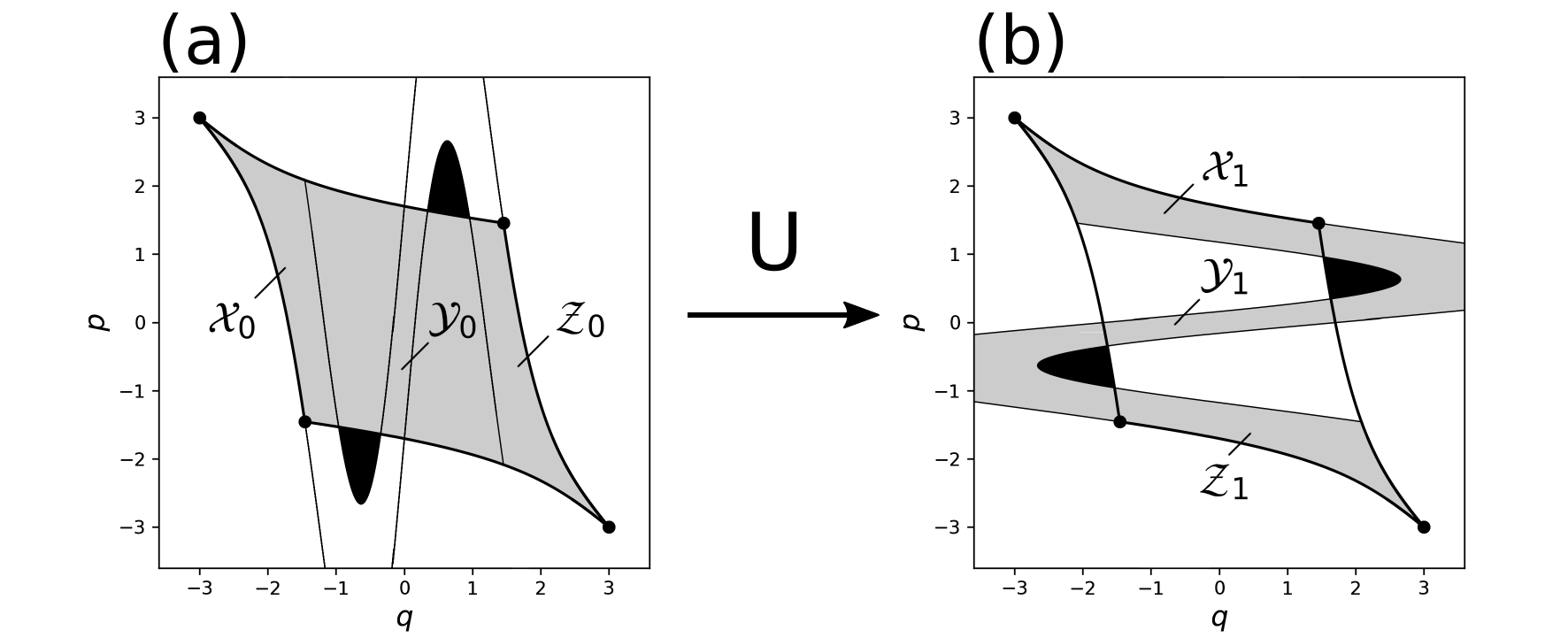}
    	\caption{
	(a) The scattering region $\mathcal{S}$(gray) and (b) its forward image $U(\mathcal{S})$(gray). 
	The intersection $\mathcal{S} \cap U(\mathcal{S})$ consists of three parts 
	$\mathcal{X}_1,\mathcal{Y}_1$ and $\mathcal{Z}_1$. Black regions 
	in (a) are mapped to the ones in (b). 
	\label{fig:Horseshoe_1}}
\end{figure} 
\begin{figure}[H]
  \centering
	\includegraphics[width = 15.0cm,bb =0 0 850 346]{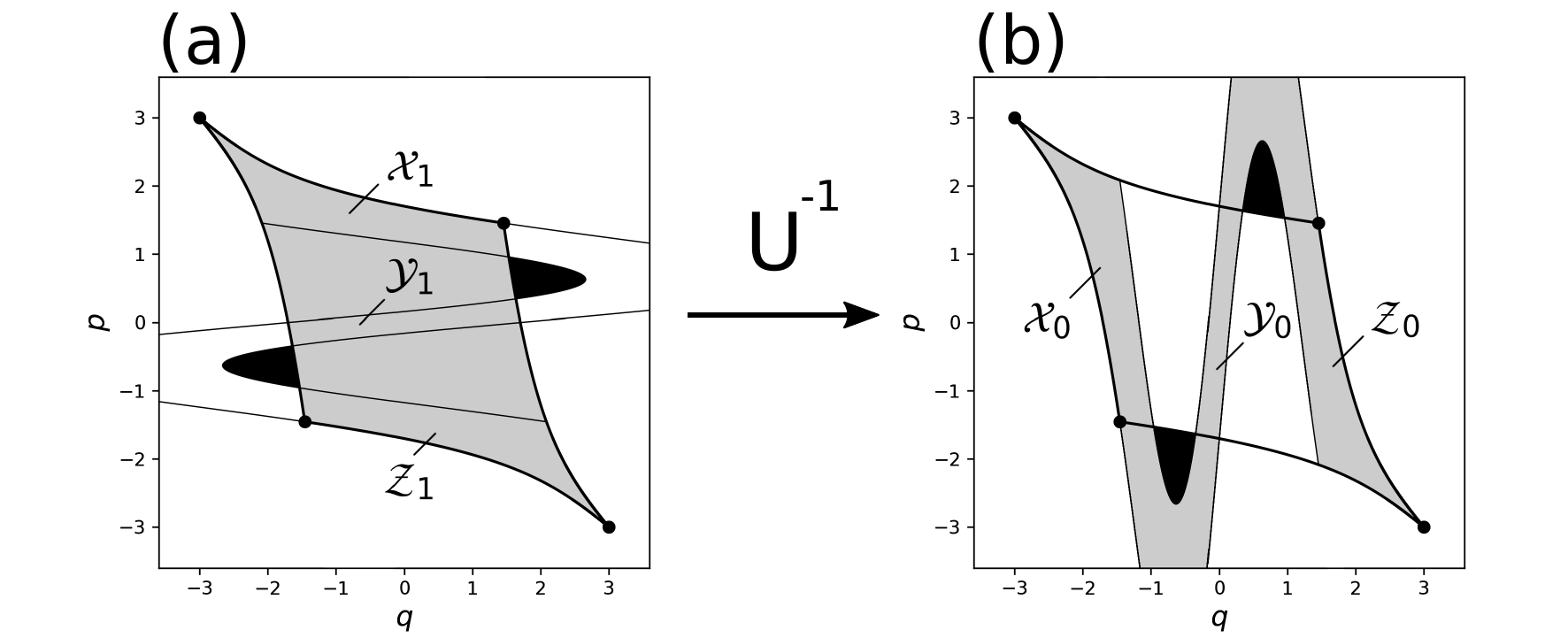}
    	\caption{
(a) The scattering region $\mathcal{S}$(gray) and (b) its backward image $U^{-1}(\mathcal{S})$(gray). 
	The intersection $\mathcal{S} \cap U^{-1}(\mathcal{S})$ consists of three parts 
	$\mathcal{X}_0,\mathcal{Y}_0$ and $\mathcal{Z}_0$. Black regions 
	in (a) are mapped to the ones in (b). 
	\label{fig:Horseshoe_2}}
\end{figure} 
\section{Topological Horseshoe condition}
\setcounter{equation}{0}
We first introduce the square region $R$ (see Fig. \ref{region_R}) enclosed by 
\begin{eqnarray}
l_1&:=& \{ (q,p) \, | \, q = -q_f, -q_f < p < q_f \} ,\\
l_2&:=& \{ (q,p) \, | \, -q_f < q < q_f,p = q_f  \} , \\
l_3&:=& \{ (q,p) \, | \, q = q_f, -q_f < p < q_f \} , \\
l_4&:=& \{ (q,p) \, | \, -q_f < q < q_f,p = -q_f \}, 
\end{eqnarray}
\begin{figure}[H]
        \centering
	\includegraphics[width = 8.0cm,bb = 0 0 461 346]{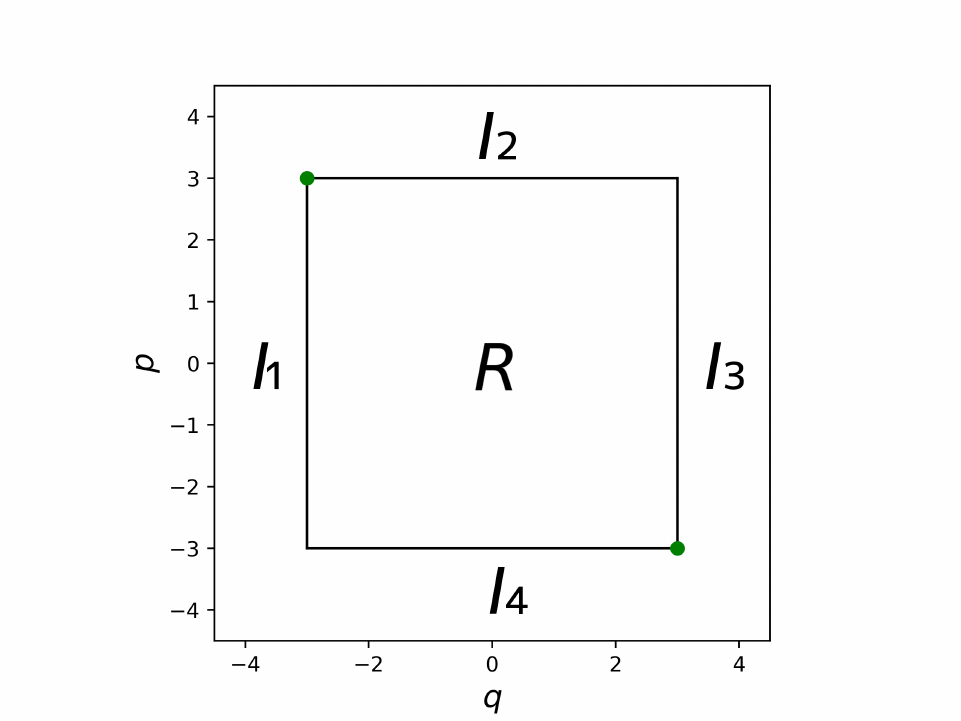}
	\caption{
	Illustration of the region $R$. The green dots represent fixed points of the map. }
	\label{region_R}
\end{figure} 
\begin{figure}[H]
  \centering
        \centering
	\includegraphics[width = 8.0cm,bb = 0 0 1024 770]{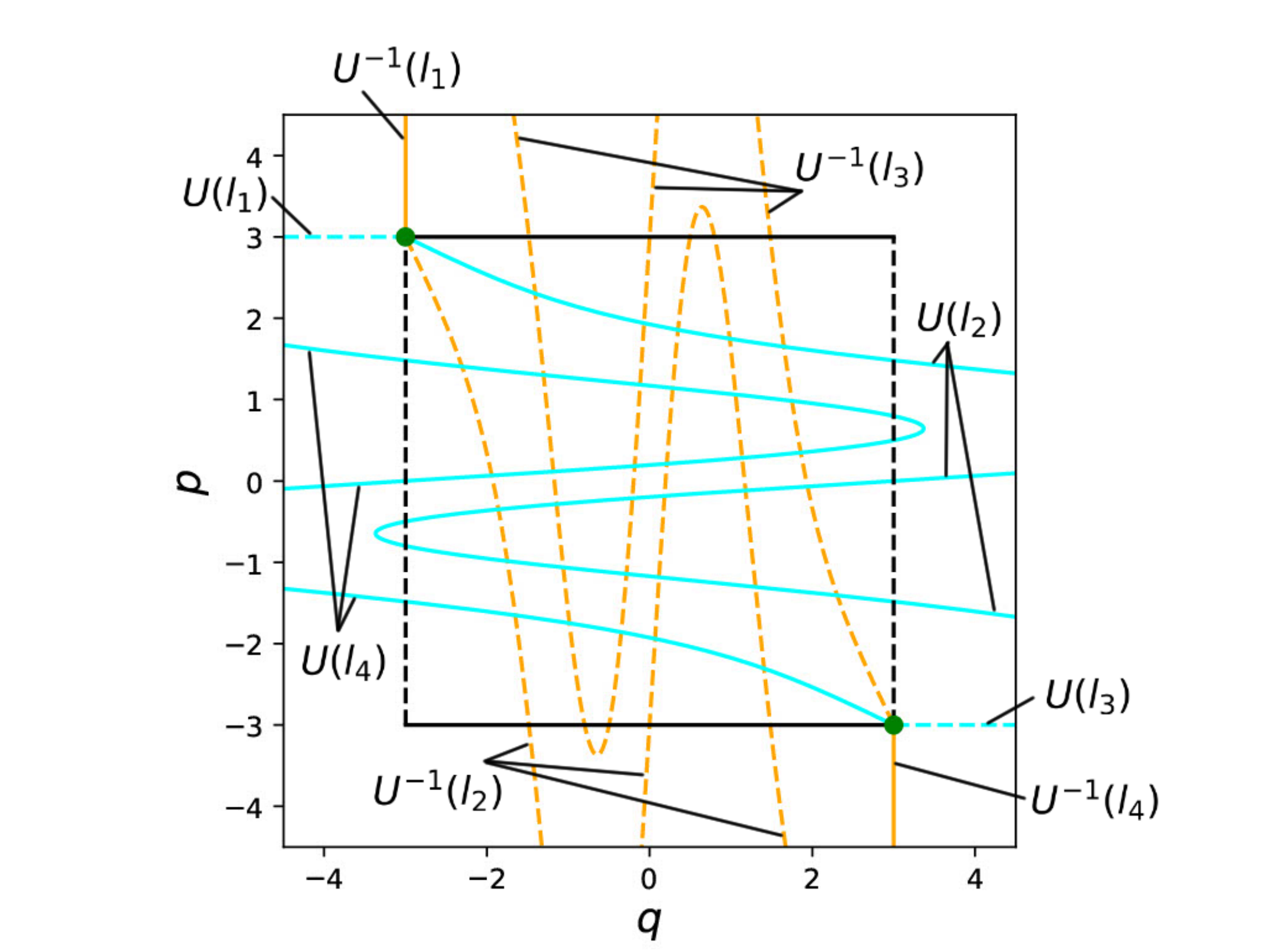}
 \caption{
	Illustration of the region $R$ (black), $U^{-1}(R)$(orange) and $U(R)$(cyan). The green dots represent fixed points of the map. 
	}
	\label{region_R_UR}
\end{figure}
The region $R$ is mapped to the region $U(R)$ (see Fig. \ref{region_R_UR}), whose boundaries 
are composed of
\begin{eqnarray}
U(l_1)&:=& \{ (q,p) \, | \, -3q_f < q < -q_f , p = q_f\} ,\\
U(l_2)&:=& \{ (q,p) \, | \, q = F(p) + q_f, -q_f < p < q_f \} , \\
U(l_3)&:=& \{ (q,p) \, | \, q_f < q < 3q_f, p = -q_f \} , \\
U(l_4)&:=& \{ (q,p) \, | \, q = F(p) - q_f, -q_f < p < q_f \}, 
\end{eqnarray}
where 
\begin{eqnarray}
F(x) = -2x + V'(x). 
\end{eqnarray}
In a similar way, the boundaries of the region ${U}^{-1}(R)$ (see Fig. \ref{region_R_UR}) are given as
\begin{eqnarray}
U^{-1}(l_1)&:=& \{ (q,p) \, | \, -q_f < q < q_f, p = F(q) - q_f\} ,\\
U^{-1}(l_2)&:=& \{ (q,p) \, | \, q = -q_f, q_f < p < 3q_f \} , \\
U^{-1}(l_3)&:=& \{ (q,p) \, | \, -q_f < q < q_f, p = F(q) + q_f \} , \\
U^{-1}(l_4)&:=& \{ (q,p) \, | \,  q = q_f, -3q_f < p < -q_f \}.
\end{eqnarray}
Due to the rotational symmetry with respect to the origin, 
the topological horseshoe is obviously realized 
if the curve $U^{-1}(l_1)$ intersect with the line $l_2$ twice for $-q_f<q<q_f$. 
Below, we derive a sufficient condition for the topological horseshoe. 
\begin{dfn}
\label{dfn:q0}
Let $q_0$ be the positive solution of
\begin{eqnarray}
\label{eq:eq_for_q3}
\frac{d}{dq} \left(  v'(q^2) +2 q^2 v''(q^2) \right) = 0
\end{eqnarray}
\end{dfn}
\noindent
{\it Remark: } 
Note that the solution for the equation (\ref{eq:eq_for_q3}) is unique and for each 
potential case we actually have
\begin{eqnarray}
q_{0}^{(G)} = \sqrt{\frac{3}{2}},\\
q_{0}^{(L)} = 1. 
\end{eqnarray}
We also note that the the function $v'(q^2) +2 q^2 v''(q^2)$ takes a local maximum value
 at $q_{0}^{(G)}$ (resp. $q_{0}^{(L)}$) for $0<q<q_f$. 
As for the profile of the function $F(q)$ defined above, 
we can show the following. 
\begin{lmm}
\label{lmm:extrema}
If the condition
\begin{eqnarray}
\label{con_mild}
q_f < q_b - q_0 
\end{eqnarray}
is satisfied respectively  for the Gaussian and Lorentzian potential case 
where 
\begin{eqnarray}
q_{0}^{(G)} = \sqrt{\frac{3}{2}},\\
q_{0}^{(L)} = 1. 
\end{eqnarray}
then 
$
F(q) - q_f < q_f 
$
holds for $-q_f < q < 0$. 
Furthermore, if the function $F(q) - q_f$ has extrema, 
there exists a unique local maximum in $0< q < q_f$. 
\end{lmm}
The proof of this lemma will given in 
\ref{app:Gaussian_mild} and \ref{app:Lorentzian_mild}.
Note that if the condition (\ref{eq:horseshoe_condition}) in 
the Proposition \ref{prp:Horseshoe} is satisfied, 
then the function $F(q) - q_f$ automatically possesses one extrema for $0 < q < q_f$.
\begin{figure}[H]
        \centering
	\includegraphics[width = 8.0cm,bb = 0 0 461 346]{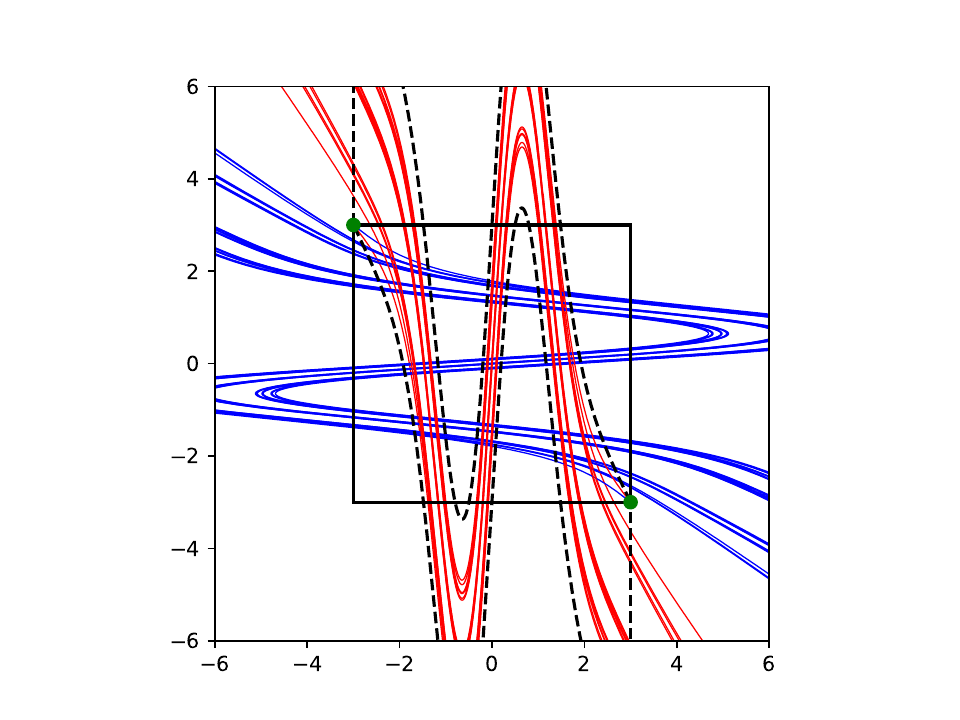}
	\caption{Illustration of the region $R$ (black solid line),$U^{-1}(R)$ (dashed line), which include the unstable(blue) and stable(red) manifold. Green dots represent fixed points of the map. }
	\label{region_R_with_manifold}
\end{figure} 
We also illustrate in Fig. \ref{region_R_with_manifold} 
that intersections of the regions $R$ and $U^{-1}(R)$ 
contain the heteroclinic points, which are given as intersections of 
stable and unstable manifolds associated with fixed points. 

\begin{dfn}
Let $q_1$ be the positive solution of 
\begin{eqnarray}
\label{eq:eq_for_q1}
\frac{ d \left( - q   v'(q^2) \right)}{dq}  
= 0, 
\end{eqnarray}
and let $M_1$ be defined as
\begin{eqnarray}
M_1 := -q_1 v'(q_1^2). 
\end{eqnarray}
\end{dfn}
\noindent
{\it Remark: } 
Note that the solution for the equation (\ref{eq:eq_for_q1}) is unique and for each 
potential case we actually have 
\begin{eqnarray}
q_{1}^{(G)} = \frac{1}{\sqrt{2}}, \quad M_1^{(G)} = \frac{1}{\sqrt{2e}}  \\
q_{1}^{(L)} = \frac{1}{\sqrt{3} }, \quad M_1^{(L)} = \frac{3\sqrt{3}}{16}
\end{eqnarray}
We also note that the function $-q   v'(q^2)$ takes the maximum value
$M_1^{(G)}$ (resp. $M_1^{(L)}$) at $q_{1}^{(G)}$ (resp. $q_{1}^{(L)}$)  and the minimum value $-M_1^{(G)}$ (resp. $-M_1^{(L)}$) at $-q_{1}^{(G)}$ (resp. $-q_{1}^{(L)}$) . 

Since $F(0)-q_f=-q_f < q_f$ and $F(q_f) = -3q_f < q_f$ hold, 
the region $R$
forms the topological horseshoe 
if there exists $\k > 0$ such that
\begin{eqnarray}
\max_{0 < q < q_f} \{ F(q)-q_f \} > q_f
\end{eqnarray}
is satisfied.  On the basis of this observation, we obtain a sufficient condition for the horseshoe:
\begin{prp}
\label{prp:Horseshoe}
If $\k$ satisfies the condition, 
\begin{eqnarray}
\label{eq:horseshoe_condition}
\k \geq \k_1 \\
\label{def:kappa1}
\k_1 := \frac{4 q_f} { \displaystyle  2M_1 - \varepsilon }, 
\end{eqnarray}
and the following conditions for $q_f$
\begin{eqnarray}
\label{eq:horseshoe_condition2}
q_1 < q_f, 
\end{eqnarray}
and
\begin{eqnarray}
\label{eq:horseshoe_condition3}
\varepsilon < \varepsilon_1 \\
\varepsilon_1 :=2M_1
\end{eqnarray}
hold, 
then the topological horseshoe is realized in the map (\ref{eq:map1}). 
\end{prp}

\begin{proof}
An explicit expression for the curve $F(q)$ is given as
\begin{eqnarray}
F (q) &=& -2q + V'(q)\nonumber \\
\fl
  &=& -2q - 2 \k q   v'(q^2) - \k \varepsilon \left( v((q-q_b)^2)  -  v((q+q_b)^2)  \right). 
\end{eqnarray}
From the form of $v(q^2)$, we easily find that
\begin{eqnarray}
\inf_{q \in \mathbb{R}} v(q^2)  =0 , \\
\sup_{q \in \mathbb{R}} v(q^2)  = 1. 
\end{eqnarray}
Since $- 1<v((q-q_b)^2) - v((q+q_b)^2)< 1$, the curve $F(q)$ is bounded as 
\begin{eqnarray}
F(q) > -2q - 2 \k q   v'(q^2) -\k \varepsilon. 
\label{eq:upper_F(q)}
\end{eqnarray}
From our assumption, the function $- q v'(q^2)$ takes the maximum at $q = q_1 < q_f$, 
which leads to
\begin{eqnarray}
\label{Prf:Prop_h_max}
\max_{0 < q < q_f} F(q)  >  -2q_f  + 2 \k M_1 - \k  \varepsilon.
\end{eqnarray}
If $\displaystyle \max_{0 < q < q_f} \{F(q)-q_f\} > q_f$ holds, 
using the lemma \ref{lmm:extrema}, 
we can say that
the curve $p = F(q)-q_f$ intersects with the line $p = q_f$ twice for $0 < q < q_f$. 
A sufficient condition for this situation is that 
the right-hand side of Eq.~(\ref{Prf:Prop_h_max}) is greater than $2q_f$, 
which is explicitly written as 
\begin{eqnarray}
\nonumber
 -4q_f  + 2 \k M_1 -\k  \varepsilon \geq 0. 
\end{eqnarray}
Together with the condition (\ref{eq:horseshoe_condition3}), 
we obtain our desired inequality (\ref{eq:horseshoe_condition}). 
\end{proof}

\noindent 
{\it Remark: } In \ref{app:Gaussian_qf_qb} and  \ref{app:Lorentzian_qf_qb}, we examine the condition (\ref{eq:horseshoe_condition3})

\section{Nonwandering set and the filtration property}
\setcounter{equation}{0}
In order to show that 
the nonwandering set $\Omega(U)$ of the system is uniformly hyperbolic, we here prove 
that the nonwandering set $\Omega(U)$ is a subset of $R \cap U^{-1}(R)$ by 
showing that the complement of $R \cap U^{-1}(R)$ is wandering. 
To this end, we introduce the following regions, 
\begin{eqnarray}
\mathcal{O}^{+} = \{(q,p)\: | \: q > q_f, p > -q \} ,  \\
\mathcal{O}^{-} = \{(q,p)\: | \: q < -q_f, p < - q \} , \\
\mathcal{I}^{+} = \{(q,p)\: | \: q > q_f, p < -q \} , \\
\mathcal{I}^{-} = \{(q,p)\: | \: q < -q_f, p > - q \}. 
\end{eqnarray}
As was shown in \cite{NS18}, 
these regions have the following properties:  
\begin{lmm} 
\label{lemma4_1}
As mentioned in Section \ref{sec:scattering_map}, 
the following holds for the potential function $V(q)$:
\begin{eqnarray}
V'(q) < 0 ~~{\rm for}~~ q > q_f, 
\end{eqnarray}
which automatically implies $V'(q) > 0$ for $q < - q_f$ because of the 
symmetry of the potential function. 
We then have the filtration property: 
\begin{description} 
\item[a)] $U(\mathcal{O}^{+}) \subset \mathcal{O}^{+}$ and $U(\mathcal{O}^{-}) \subset \mathcal{O}^{-}$. 
\item[b)] $q_n \in \mathcal{O}^{+}$ is strictly increasing, 
and $q_n \in \mathcal{O}^{-}$ is strictly decreasing under the forward iteration of the map $U$. 
\item[c)] $U^{-1}(\mathcal{I}^{+}) \subset \mathcal{I}^{+}$ and $U^{-1}(\mathcal{I}^{-}) \subset \mathcal{I}^{-}$.
\item[d)]$q_n \in \mathcal{I}^{+}$ is strictly increasing, 
and $q_n \in \mathcal{I}^{-}$ is strictly decreasing under the backward iteration of the map $U$. 
\end{description}
\label{prp:filtration}
\end{lmm}
\begin{proof}
Below we give the proof provided the condition (\ref{eq:cond_outgoing}) holds.
From the condition (\ref{eq:cond_outgoing}),
for $(q_{n},p_{n}) \in \mathcal{O}^{+}$ 
\begin{eqnarray}
q_{n+1} +p_{n+1}= q_{n} + p_{n} - V'(q_n) >  q_{n} + p_{n} =: \Delta > 0 
\end{eqnarray}
holds, and 
\begin{eqnarray}
q_{n+1} =2 q_{n} + p_{n} - V'(q_n) >  q_{n}  + (q_{n} + p_{n}) = q_{n}  + \Delta > q_f, 
\end{eqnarray}
which shows that $U(\mathcal{O}^{+}) \subset \mathcal{O}^{+}$. In addition, conbining both equations iteratively we have:
\begin{eqnarray}
q_{n+m} > q_{n}  + m \Delta. 
\end{eqnarray}
Using the symmetry, the statements for $\mathcal{O}^{-},\mathcal{I}^{+}$ and $\mathcal{I}^{-}$
immediately follow.
\end{proof}

\begin{figure}[H]
        \centering
	\includegraphics[width = 10.0cm,bb =  0 0 461 346]{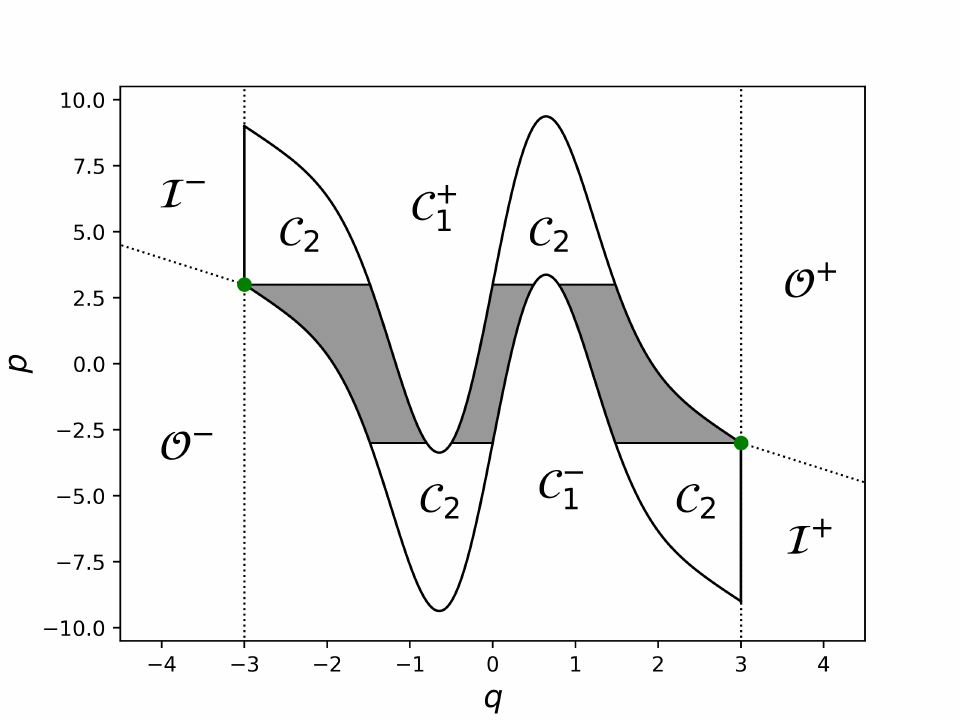}
	\caption{
	The set $R \cap U^{-1}(R)$ (gray) and its complement.
	The complement is a union of the sets 
	 $\mathcal{C}_{1}^{\pm},\mathcal{C}_{2},\mathcal{O}^{\pm}$ and $\mathcal{I}^{\pm}$. 
	The region $\mathcal{C}_{1}^{+}$ (resp.  $\mathcal{C}_{1}^{-}$) is mapped into the region $\mathcal{O}^{+}$ (resp. $\mathcal{O}^{-}$) under forward iteration. The region $\mathcal{C}_{2}$ is mapped into the region $(R \cap U^{-1}(R)) \cup \mathcal{C}_{1}^{+} \cup \mathcal{C}_{1}^{-}$ under the forward iteration, meaning that the points contained in the set $\mathcal{C}_{2}$ either stay in $R \cap U^{-1}(R)$ or go out to $\mathcal{O}^{\pm}$ under more than one-step the forward iterations. 
}
	\label{fig:filtration1}
\end{figure} 
\begin{figure}[H]
        \centering
\includegraphics[width = 10.0cm,bb =  0 0 461 346]{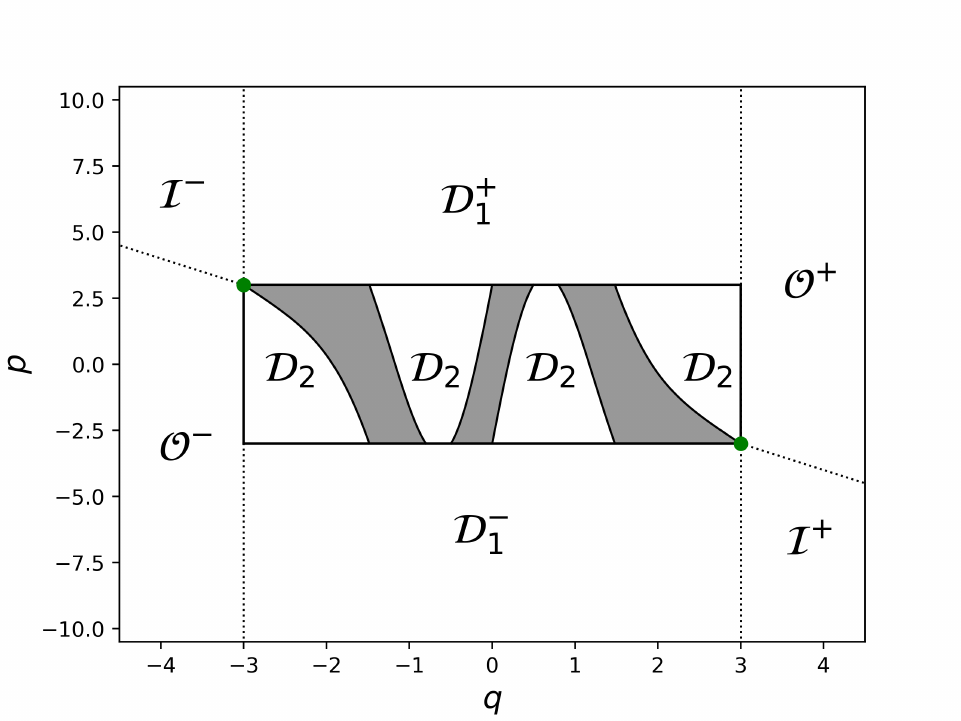}
\caption{
	The set $R \cap U^{-1}(R)$ (gray) and its complement.
	The complement is a union of the sets 
	$\mathcal{D}_{1}^{\pm},\mathcal{D}_{2},\mathcal{O}^{\pm}$ and $\mathcal{I}^{\pm}$. 
	The region $\mathcal{D}_{1}^{+}$ (resp.  $\mathcal{D}_{1}^{-}$) is mapped into the region $\mathcal{I}^{+}$ (resp. $\mathcal{I}^{-}$) under the backward iteration. The region $\mathcal{D}_{2}$ is mapped into the region $(R \cap U^{-1}(R)) \cup \mathcal{D}_{1}^{+} \cup \mathcal{D}_{1}^{-}$ under the backward iteration, meaning that the points contained in the set $\mathcal{D}_{2}$ either stay in $R \cap U^{-1}(R)$ or go out to $\mathcal{I}^{\pm}$ under more then one-step backward iterations. 
}
	\label{fig:filtration2}
\end{figure} 
\par
Next we consider the behavior of the internal region $\{ (q,p) \: | -q_f < q < q_f \}$ under the iteration. 
To this end we focus on the forward and backward image of the complement 
of $R \cap U^{-1}(R)$, respectively. 
As shown in Fig. \ref{fig:filtration1}, we introduce subsets ($\mathcal{C}_{1}^{\pm},\mathcal{C}_{2})$ 
as the forward image of the complement of $R \cap U^{-1}(R)$: 
\begin{eqnarray}
\mathcal{C}_{1}^{+} = \{(q,p)\: | \: -q_f < q < q_f, p > F(q) +  q_f \} , \\
\mathcal{C}_{1}^{-} = \{(q,p)\: | \: -q_f < q < q_f, p < F(q) - q_f \} , \\
\mathcal{C}_{2} = U^{-1}(R) \backslash (R \cap U^{-1}(R)). 
\end{eqnarray}
Since $\mathcal{C}_{2} \subset U^{-1}(R)$, $U(\mathcal{C}_{2}) \subset R $ holds, $U(\mathcal{C}_{2}) \subset R \subset (R \cap U^{-1}(R))\cup \mathcal{C}_{1}^{+} \cup \mathcal{C}_{1}^{-} $ 
follows. 
As for $\mathcal{C}_{1}^{\pm}$, we have the following lemma:
\begin{lmm}
\label{lemma4_2}
If the condition (\ref{eq:cond_outgoing}) is satisfied, then 
\begin{eqnarray}
U(\mathcal{C}_1^{+}) \subset \mathcal{O}^{+} \quad and \quad U(\mathcal{C}_1^{-}) \subset \mathcal{O}^{-}. 
\end{eqnarray}
\end{lmm}
\begin{proof}
For $(q_n,p_n) \in \mathcal{C}_1^{+}$, we have 
\begin{eqnarray}
q_{n+1} = 2q_n + p_n - V'(q_n) > 2q_n + F(q_n) + q_f - V'(q_n)  =  q_f,  \quad 
\end{eqnarray}
and 
\begin{eqnarray}
p_{n+1} = -q_n. 
\end{eqnarray}
Therefore, since $-q_f < q_n < q_f$, $-q_f < p_{n+1} < q_f$ holds. 
This implies $(q_{n+1},p_{n+1}) \in \mathcal{O}^{+}$. 
One similarly shows that $U(\mathcal{C}_1^{-}) \subset \mathcal{O}^{-}$.
\end{proof}
In a similar way, Fig. \ref{fig:filtration2} illustrates subsets ($\mathcal{D}_{1}^{\pm},\mathcal{D}_{2}$), 
which are introduced as the backward image of the complement of $R \cap U^{-1}(R)$: 
\begin{eqnarray}
\mathcal{D}_{1}^{+} = \{(q,p)\: | \: -q_f < q < q_f, p > q_f \},  \\
\mathcal{D}_{1}^{-} = \{(q,p)\: | \: -q_f < q < q_f, p < - q_f \},  \\
\mathcal{D}_{2} = R \backslash (R \cap U^{-1}(R)). 
\end{eqnarray}
Since $\mathcal{D}_{2} \subset R$ and $U^{-1}(\mathcal{D}_{2}) \subset U^{-1}(R) $ hold, 
$U^{-1}(\mathcal{D}_{2}) \subset (R \cap U^{-1}(R))\cup \mathcal{D}_{1}^{+} \cup \mathcal{D}_{1}^{-}$ follows. 
As for $\mathcal{D}_{1}^{\pm}$, we have the lemma:
\begin{lmm}
\label{lemma4_3}
If the condition (\ref{eq:cond_outgoing}) is satisfied, then 
\begin{eqnarray}
U^{-1}(\mathcal{D}_1^{+}) \subset \mathcal{I}^{-} \quad and \quad U^{-1}(\mathcal{D}_1^{-}) \subset \mathcal{I}^{+} . 
\end{eqnarray}
\end{lmm}
\begin{proof}
For $(q_n,p_n) \in \mathcal{D}_1^{+}$, we have 
\begin{eqnarray}
q_{n-1} = - p_n < -q_f < q_n. 
\end{eqnarray}
Combining this with the condition (\ref{eq:cond_outgoing}), this leads to 
\begin{eqnarray}
p_{n-1} &=& q_n + 2p_n + V'(q_{n-1})\nonumber \\
&=& q_n - 2q_{n-1}+ V'(q_{n-1}) \nonumber \\
&>& q_{n-1} - 2q_{n-1}+ V'(q_{n-1})  > -q_{n-1}.
\nonumber
\end{eqnarray}
This implies $(q_{n-1},p_{n-1}) \in \mathcal{I}^{-}$. 
One similarly shows that $U^{-1}(\mathcal{D}_1^{-}) \subset \mathcal{I}^{+}$.
\end{proof}
Considering the behavior of the complement of  $R \cap U^{-1}(R)$ under respectively forward and backward iterations of $U$, 
we reach the following. 
\begin{prp}
\label{prp:nonwandering}
If the condition (\ref{eq:cond_outgoing}) is satisfied, 
$\Omega(U) \subset R \cap U^{-1}(R)$ holds. 
\end{prp}
\begin{proof}
From Lemma \ref{lemma4_1}, Lemma \ref{lemma4_2} and Lemma \ref{lemma4_3},
we can say that the complement of the set $R \cap U^{-1}(R)$ is wandering. 
\end{proof}

\section{Conley-Moser Conditions}
\setcounter{equation}{0}
Following the argument presented in the book of Wiggins \cite{Wiggins}, we 
here introduce a sufficient condition, the so-called Conley-Moser condition, 
for the invertible 2-dimensional map to have 
the conjugation with the symbolic dynamics, which is given by a full shift on 
a finite number of symbolic space in general. 
Below we first provide two conditions, referred to as the Assumptions 1 and 2, 
and then derive a sufficient condition for our map $U$, under which 
the Assumptions 1 and 2 are both fulfilled. 
We will call the sufficient condition derived for the map $U$ 
in such a way the {\it sector condition} hereafter. 
\subsection{Outline of the Conley-Moser conditions\cite{Wiggins}}
\begin{dfn}[Vertical and horizontal curve]
A $\mu_v$-verical curve is the graph of a function $q = v(p)$ which satisfies
\begin{eqnarray}
\fl
\quad -q_f \leq v(p) \leq q_f,\quad|v(p_1)-v(p_2)|\leq \mu_v |p_1 - p_2|\quad for\: -q_f\leq p_1,p_2 \leq q_f.
\end{eqnarray}
Similarly, a $\mu_h$-horizontal curve is the graph of a function $p = h(q)$ which satisfies
\begin{eqnarray}
\fl
\quad -q_f \leq h(q) \leq q_f,\quad|h(q_1)-h(q_2)|\leq \mu_h |q_1 - q_2|\quad for\: -q_f\leq q_1,q_2 \leq q_f.
\end{eqnarray}
\end{dfn}
\begin{dfn}[Vertical and horizontal strip]
Given two nonintersecting $\mu_v$-vertical curves $v_1(p) <v_2(p) (p \in [-q_f,q_f]$),
a $\mu_v$-vertical strip is the region $V$ which satisfies
\begin{eqnarray}
V := \{(q,p)|\: v_1(p)<q < v_2(p) ,\: -q_f<p<q_f\}.
\end{eqnarray}
Similarly, given two nonintersecting $\mu_h$-horizontal curves $h_1(q) <h_2(q) (q \in [-q_f,q_f]$),
a $\mu_h$-horizontal strip is the region $H$ which satisfies
\begin{eqnarray}
H := \{(q,p) \: |\: h_1(q)<p<h_2(q),\: -q_f<q<q_f\}.
\end{eqnarray}
We further define
\begin{eqnarray}
H_{ij} := V_i \cap U(V_j)
\end{eqnarray}
and
\begin{eqnarray}
V_{ji} := U^{-1}(V_i)\cap V_j = U^{-1}(H_{ji})
\end{eqnarray}
for $i,j \in S$ where $S = \{ 1,...,N \}$. We then introduce the unions as
\begin{eqnarray}
\mathcal{H} := \bigcup_{i,j \in S} H_{ij},\quad \mathcal{V} := \bigcup_{i,j \in S} V_{ji}.
\end{eqnarray}
\end{dfn}
\begin{dfn}[Definition of the cone field (sector bundle)]
For any point $(q_n,p_n)$ in the phase space, and the associated tangent space  $(\xi_n, \eta_n)$, we define the $\mu_v$-stable cone field along p-axis at $(q_n,p_n)$ as follows: 
\begin{eqnarray}
\mathcal{S}_{(q_n,p_n)}^{s} := \{(\xi_n,\eta_n) \:|\: |\xi_n| \leq \mu_v |\eta_n| \}.
\end{eqnarray}
Similarly, the $\mu_h$-unstable cone field along q-axis at $(q_n,p_n)$ is defined as follows: 
\begin{eqnarray}
\mathcal{S}_{(q_n,p_n)}^{u} := \{(\xi_n,\eta_n) \:|\:  |\eta_n| \leq \mu_h |\xi_n| \}.
\end{eqnarray}
Then, we introduce the $\mu_v$-stable and $\mu_h$-unstable cone fields for $\mathcal{H}$ and $\mathcal{V}$ in the following way:
\begin{eqnarray}
\mathcal{S}_{\mathcal{H}}^{s} := \bigcup_{(q_n,p_n)\in \mathcal{H}} \mathcal{S}_{(q_n,p_n)}^{s},\quad 
\mathcal{S}_{\mathcal{V}}^{s} := \bigcup_{(q_n,p_n)\in \mathcal{V}} \mathcal{S}_{(q_n,p_n)}^{s}, \\
\mathcal{S}_{\mathcal{H}}^{u} := \bigcup_{(q_n,p_n)\in \mathcal{H}} \mathcal{S}_{(q_n,p_n)}^{u},\quad 
\mathcal{S}_{\mathcal{V}}^{u} := \bigcup_{(q_n,p_n)\in \mathcal{V}} \mathcal{S}_{(q_n,p_n)}^{u}.
\end{eqnarray}
\end{dfn}
\begin{thm}[Wiggins \cite{Wiggins}] 
\label{Wiggins} 
If the map $U$ satisfies following assumptions,
\begin{description} 
\setlength{\leftskip}{1.0cm}
  \item[Assumption 1]\mbox{}\\ 
  $0\leq \mu_v \mu_h < 1$ and $U$ maps $V_i$ homeomorphically onto $H_i$, 
  that is 
  $U(V_i) = H_i$ for $i = 1,\cdots,N$. Moreover, the horizontal boundaries of $V_i$ map to the horizontal boundaries of $H_i$ and the vertical boundaries of $V_i$ map to the vertical boundaries of $H_i$. 
  \item[Assumption 2] \mbox{}\\
 $DU(\mathcal{S}_{\mathcal{V}}^{u}) \subset \mathcal{S}_{\mathcal{H}}^{u}$ and $DU^{-1}(\mathcal{S}_{\mathcal{H}}^{s}) \subset \mathcal{S}_{\mathcal{V}}^{s}$.
Moreover, if $(\xi_n,\eta_n)\in \mathcal{S}_{(q_n,p_n)}^{u}$ and $(\xi_{n+1},\eta_{n+1}) \in  \mathcal{S}_{(q_{n+1},p_{n+1})}^{u}$ where $(\xi_{n+1},\eta_{n+1}) =  DU_{(q_n,p_n)}((\xi_n,\eta_n))$,
then 
\begin{eqnarray}
|\xi_{n+1}| \geq \frac{1}{\lambda}|\xi_n|.
\label{ineq:Assum2_a}
\end{eqnarray}
Similarly, if $(\xi_n,\eta_n)\in \mathcal{S}_{(q_n,p_n)}^{s}$ and $(\xi_{n-1},\eta_{n-1}) \in  \mathcal{S}_{(q_{n-1},p_{n-1})}^{s}$ where $ (\xi_{n-1},\eta_{n-1})= DU^{-1}_{(q_n,p_n)}((\xi_n,\eta_n))$, then
\begin{eqnarray}
|\eta_{n-1}| \geq \frac{1}{\lambda}|\eta_n|, 
\label{ineq:Assum2_b}
\end{eqnarray}
where $0 < \lambda < 1-\mu_h \mu_v$.
\end{description}
then, the following statements hold
\begin{description} 
\setlength{\leftskip}{1.0cm}
  \item[i)] $\Omega(U)$ is a non-empty compact invariant Cantor set, on which 
  the map $U$ is topologically conjugate to N-shift map.
  \item[ii)] $\Omega(U)$ is a uniformly hyperbolic invariant set of the map $U$.
\end{description}
\end{thm}

\subsection{Application to the scattering map}
From Theorem \ref{Wiggins}, if our map $U$ satisfies the Assumptions 1 and 2, 
uniform hyperbolicity follows. 
In this subsection, we first derive a sufficient condition (sector condition) 
leading to the Assumptions 1 and 2.
Then we show that there exist vertical and horizontal curves introduced above 
if the derived sufficient condition is satisfied. 

To this end,  
we denote the $q$-coordinates of some of the intersecting points, as indicated in Fig. \ref{fig:nu_def},  of 
the curves $U^{-1}(l_1)$ and $U^{-1}(l_3)$ with 
the line $l_2 =\{(q,p) | -q_f < q < q_f, p = q_f \}$ by 
$0 < \nu_1 < \nu_2< \nu_3< q_f$. 
\begin{figure}[H]
        \centering
	\includegraphics[width = 10.0cm,bb =  0 0 461 346]{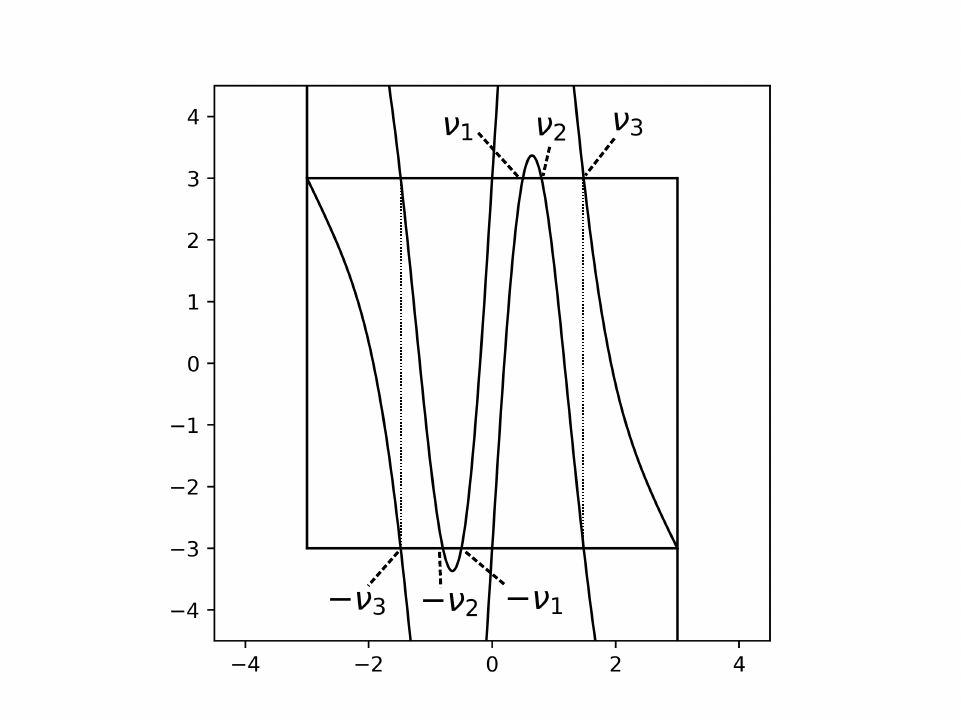}
	\caption{Intersections $\nu_i(i = 1,2,3)$ of 
the curves $U^{-1}(l_1)$ and $U^{-1}(l_3)$ with
the line $l_2 =\{(q.p) | -q_f < q < q_f, p = q_f \}$. 
}
	\label{fig:nu_def}
\end{figure} 
We further introduce the intervals 
\begin{eqnarray}
\label{vertical_interval}
I_V := \{q \in \mathbb{R} | -q_f \leq q \leq -\nu_2, -\nu_1 \leq q \leq \nu_1, \nu_2 \leq q \leq q_f \},  \\
\label{horitonal_interval}
I_H := \{p \in \mathbb{R} | -q_f \leq p \leq -\nu_2, -\nu_1 \leq p \leq \nu_1, \nu_2 \leq p \leq q_f \}. 
\end{eqnarray}
\begin{prp}
\label{Prp:Sector_condition}
For $\displaystyle \tilde{\mu} := \mu + \frac{1}{\mu}$ with $0 \leq \mu < 1$, 
if the following conditions for $F(q)$ and $F(p)$, 
\begin{eqnarray}
\label{sector_cond}
\left|\frac{dF(q)}{dq}\right| \geq \tilde{\mu}\:~for ~\: q \in I_V ~~and \quad \left|
\frac{dF(p)}{dp}\right| \geq \tilde{\mu}\: ~for ~ \:p \in I_H, 
\end{eqnarray}
are all fulfilled, then the Assumptions 1 and 2 hold.  
\end{prp}
The conditions (\ref{sector_cond}) and (\ref{sector_cond_in}) will be 
referred to as the {\it sector conditions}. 
\begin{dfn}
We introduce the following functions associated with $F^{-1}(q)$: 
\begin{eqnarray}
&& v_1^{-1}(q) := F(q)- q_f ~~{\rm for}~~ q \in [-q_f, - \nu_3] ,\\
&& v_2^{-1}(q) := F(q)+ q_f ~~{\rm for}~~ q \in [-\nu_3, -\nu_2] ,\\
&& v_3^{-1}(q) := F(q)+ q_f ~~{\rm for}~~ q \in [-\nu_1, 0],\\
&& v_4^{-1}(q) := F(q)- q_f ~~{\rm for}~~ q \in [0, \nu_1] ,\\
&& v_5^{-1}(q) := F(q)- q_f ~~{\rm for}~~ q \in [\nu_2, \nu_3] ,\\
&& v_6^{-1}(q) := F(q)+ q_f ~~{\rm for}~~ q \in [\nu_3, q_f] .
\end{eqnarray}
In a similar way, we introduce 
\begin{eqnarray}
&& h_1^{-1}(p) := F(p)+ q_f  ~~{\rm for}~~ p \in [\nu_3, q_f],\\
&& h_2^{-1}(p) := F(p)- q_f  ~~{\rm for}~~ p \in [\nu_2, \nu_3] ,\\
&& h_3^{-1}(p) := F(p)- q_f  ~~{\rm for}~~ p \in [0, \nu_1] ,\\
&& h_4^{-1}(p) := F(p)+ q_f  ~~{\rm for}~~ p \in [-\nu_1, 0] ,\\
&& h_5^{-1}(p) := F(p)+ q_f  ~~{\rm for}~~  p \in [-\nu_3, -\nu_2] ,\\
&& h_6^{-1}(p) := F(p)- q_f  ~~{\rm for}~~ p \in [-q_f, - \nu_3] .
\end{eqnarray}
\end{dfn}
\begin{dfn}
We then consider the regions in figure 12 whose boundaries are 
given by the functions $v_i(p)~ (i = 1,\cdots,6)$: 
\begin{eqnarray}
V_1 := \{(q,p) \,|\, v_1(p)<q<v_2(p), -q_f < p< q_f\}, \\
V_2 := \{(q,p) \,|\, v_3(p)<q<v_4(p), -q_f < p< q_f\}, \\
V_3 := \{(q,p) \,|\, v_5(p)<q<v_6(p), -q_f < p< q_f\}. 
\end{eqnarray}
Similarly, using the functions $h_i(p)~ (i = 1,...,6)$ we introduce in figure 13 
\begin{eqnarray}
H_1 := \{(q,p) \,|\, h_1(q)>p>h_2(q), -q_f < q< q_f\},  \\
H_2 := \{(q,p) \,|\, h_3(q)>p>h_4(q), -q_f < q< q_f\}, \\
H_3 := \{(q,p) \,|\, h_5(q)>p>h_6(q), -q_f < q< q_f\}. 
\end{eqnarray}
\end{dfn}

\begin{lmm}
If the condition (\ref{sector_cond}) holds, 
then the curves $q = v_i(p)~(i = 1,\cdots,6)$ $($resp. $p = h_i(q)~(i = 1,\cdots, 6))$
become horizontal curves (resp. vertical curves) 
with $\mu_v = \mu$ $($resp. $\mu_h = \mu)$. 
In addition, the domains $V_j~(j = 1,2,3)$ $($resp. $H_j~(j = 1,2,3))$ 
form the horizontal strips. 
\end{lmm}
\begin{proof}
First note that the condition (\ref{sector_cond}) is equivalent to 
the condition,  
\begin{eqnarray}
\fl
\label{sector_cond_in}
\left|\frac{dF^{-1}(p)}{dp}\right| \leq \frac{1}{\tilde{\mu}}\: ~for ~ -q_f < p < q_f ~~
{\rm and} \quad \left|\frac{dF^{-1}(q)}{dq}\right| \leq \frac{1}{\tilde{\mu}}\: ~for ~ -q_f < q < q_f,   
\end{eqnarray}
Below, we only examine the case of the function $v_1(p)$ 
since the same argument applies to the other cases, $v_i(p)(i = 2,\cdots 6)$. 
For $-q_f < p_1 < p_2<q_f$, we obtain 
\begin{eqnarray}
\frac{\left| v_1(p_1) - v_1(p_2)\right|}{|p_1-p_2|} =\frac{\left| F^{-1}(p_1 +q_f) - F^{-1}(p_2 +q_f)\right|}{|p_1-p_2|} . 
\end{eqnarray}
From the mean-value theorem, there exists $p_1 < p_c < p_2$ such that
\begin{eqnarray}
\frac{\left| F^{-1}(p_1 +q_f) - F^{-1}(p_2 +q_f)\right|}{|p_1-p_2|} =\left|\left. \frac{dF^{-1}(p +q_f)}{dp}\right|_{p=p_c}\right|. 
\end{eqnarray}
From the inequality (\ref{sector_cond_in}), 
\begin{eqnarray}
\left|\left. \frac{dF^{-1}(p +q_f)}{dp}\right|_{p=p_c}\right| = \left|\left. \frac{dF^{-1}(p)}{dp}\right|_{p=p_c}\right| \leq \frac{1}{\tilde{\mu}}
\end{eqnarray}
holds. Since we have assumed that 
\begin{eqnarray}
\frac{1}{\tilde{\mu}} < \mu, 
\end{eqnarray}
we have 
\begin{eqnarray}
\frac{\left| v_1(p_1) - v_1(p_2)\right|}{|p_1-p_2|}  < \mu , 
\end{eqnarray}
which implies that $v_1(q)$ is a vertical curve. 
This immediately leads us to the fact 
that $V_i$ (resp. $H_i$) forms a vertical (resp. horizontal) strip. 
\end{proof}
\begin{lmm}
If the condition (\ref{sector_cond}) holds, then the Assumption 1 is satisfied. 
\end{lmm}
\begin{proof}
Since $\mu_v = \mu_h = \mu$ and $ 0 \leq \mu < 1$ are assumed, 
we obtain 
\begin{eqnarray}
0 \leq \mu_v \mu_h < 1. 
\end{eqnarray}
As for the boundary of the region $V_1$, 
the upper horizontal part of the boundary,
\begin{eqnarray*}
p = q_f ,\quad -q_f < q < -\nu_3, 
\end{eqnarray*}
is mapped to the curve 
\begin{eqnarray}
q = F(p) + q_f ,\quad \nu_3 < p < q_f. 
\end{eqnarray}
This is the upper horizontal part of the boundary for $H_1$. 
In a similar manner, other parts of the boundary for $V_1$ are also mapped 
to the corresponding parts of the boundary for $H_1$. 
We can show the same behavior for $V_2$ and $V_3$ as well. 
\end{proof}
\begin{figure}[H]
        \centering
	\includegraphics[width = 10.0cm,bb =  0 0 461 346]{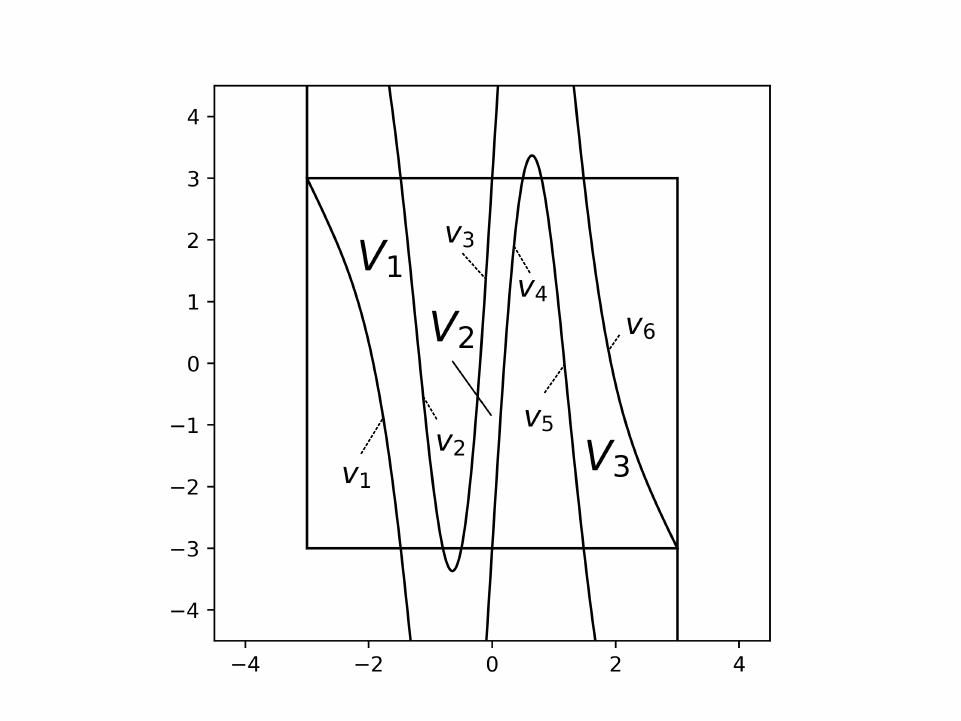}
	\caption{The vertical strips $V_1,V_2$ and $V_3$ and the curves $v_i(i = 1,\cdots ,6)$. 
}
	\label{fig:V_strip}
\end{figure} 
\begin{figure}[H]
        \centering
\includegraphics[width = 10.0cm,bb =  0 0 461 346]{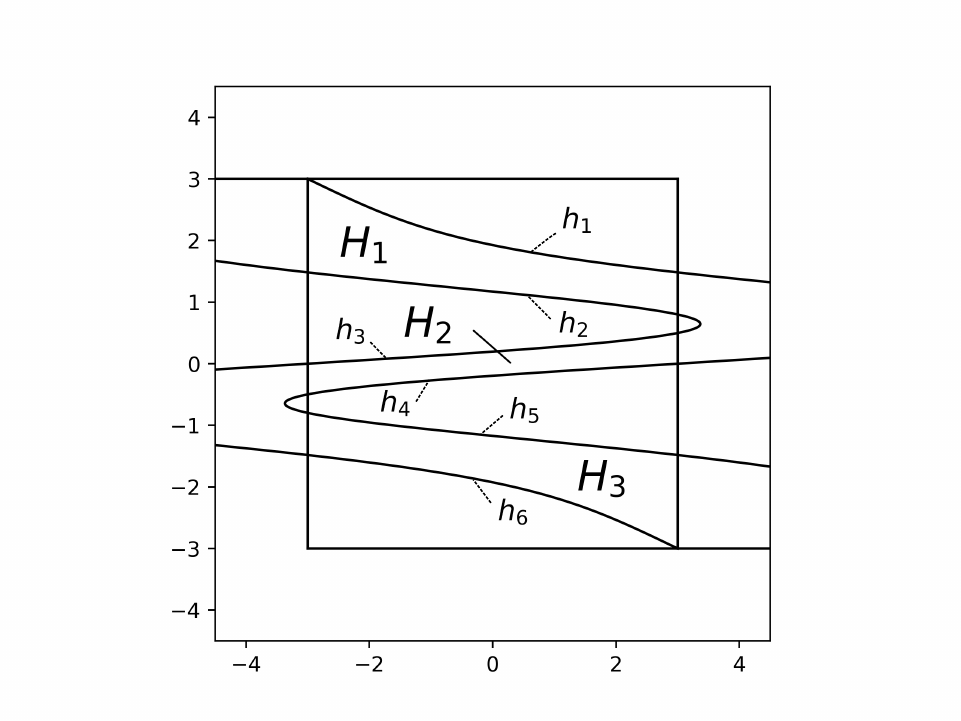}
	\caption{Horizontal strips $H_1,H_2$ and $H_3$ and the curves $h_i(i = 1,\cdots ,6)$. 
}
	\label{fig:H_strip}
\end{figure} 
In order to consider the Assumption 2, which concerns the cone fields in the tangent space, 
we here present an explicit form of the tangent map for our map (\ref{eq:Map})  with using the function $F(q)$. 
\begin{eqnarray}
DU_{(q_n,p_n)}  &=&
\left(
    \begin{array}{cc}
       2 - V''(q_{n})& 1 \\
      -1& 0
    \end{array}
  \right) 
\nonumber   \\&=& \left(
    \begin{array}{cc}
      -F'(q_{n}) & 1 \\
      - 1 & 0
    \end{array}
  \right).  \\
   \nonumber\\
DU_{(q_n,p_n)}^{-1}  &=& \left(
    \begin{array}{cc}
      0 & -1 \\
      1 & -F'(q_{n})
    \end{array}
  \right) . 
\end{eqnarray}
We remark that $F'(q) = F'(-q)$ holds for which will be used in the subsequent arguments. 
\begin{lmm}
If the condition (\ref{sector_cond}) holds, 
then the Assumption 2 is satisfied. 
\end{lmm}
\begin{proof}
As for the (\ref{ineq:Assum2_a}) in Assumption 2, 
we find that 
\begin{eqnarray*}
|\eta_{n+1}| &=& |-\xi_n| = |\xi_n| ,  \nonumber 
\end{eqnarray*}
and also 
\begin{eqnarray*}
|\xi_{n+1}| &=& |-F'(q_n)\xi_{n} + \eta_{n}| \nonumber \\ 
&\geq& |F'(q_n)\xi_{n}| - |\eta_{n}| \nonumber \\
&\geq& |F'(q_n)||\xi_{n}| - \mu|\xi_{n}| \nonumber \\
&\geq& \frac{1}{\mu}|\xi_{n}| = \frac{1}{\mu}|\eta_{n+1}|. 
\end{eqnarray*}
This implies $(\xi_{n+1},\eta_{n+1})\in \mathcal{S}_{(q_{n+1},p_{n+1})}^{u}$. 
In this situation, taking $\lambda$ satisfying the condition 
$\mu \leq \lambda < 1-\mu^2$, we can show that 
\begin{eqnarray}
|\xi_{n+1}| \geq \frac{1}{\mu}|\xi_{n}| \geq \frac{1}{\lambda}|\xi_{n}|. 
\end{eqnarray}
Notice here that $\lambda$ satisfies $0 < \lambda < 1-\mu_h \mu_v$. 
A similar argument follows for (\ref{ineq:Assum2_b}) in Assumption 2. 
\end{proof}

\section{Sufficient Conditions for sector condition}
\setcounter{equation}{0}
\label{sec:Sufficient_sector}

In this section we derive sufficient conditions for the parameter $\k$ leading to the 
sector condition (\ref{sector_cond}).
Because of the symmetry with respect to the line $p = -q$, 
it is enough to examine the sector condition only for 
the vertical strip $V_1,V_2$ and $V_3$, including the boundaries. 

Our strategy is composed of three steps:
We first check the sector condition for a certain region containing 
the interval $I_V$ introduced in (\ref{vertical_interval}) 
since it is hard to write down explicitly the coordinates $\nu_i$ which 
specify the edges of $I_V$. Next we examine the sector condition 
by replacing the function $F'(q)$ by another one because the function $F'(q)$ is also 
difficult to be controlled.  Finally, we collect all the conditions obtained in each step and 
provide the final result for the Gaussian and Lorentzian cases separately. 

\subsection{Preliminary for the division of the phase space}\begin{dfn}
Let $q_2$ be the positive solution of 
\begin{eqnarray}
\label{eq:eq_for_q2}
\frac{d (2q v''(q^2)) }{dq}= 0 , 
\end{eqnarray}
and let $M_2$ be defined as
\begin{eqnarray}
M_2 :=  2q_2 v''(q_2^2). 
\end{eqnarray}
\end{dfn}
\noindent
{\it Remark: } 
The solution for the equation (\ref{eq:eq_for_q2}) is unique and for each 
potential case we actually have
\begin{eqnarray}
q_{2}^{(G)} = \frac{1}{\sqrt{2}}, \quad M_2^{(G)} = \sqrt{\frac{2}{e}},   \\
q_{2}^{(L)} = \frac{1}{\sqrt{5}}, \quad M_2^{(L)} = \frac{25\sqrt{5}}{54}. 
\end{eqnarray}
We also note that the the function $2q v''(q^2)$ takes the maximum value
$M_2^{(G)}$ (resp. $M_2^{(L)}$) at $q_{2}^{(G)}$ (resp. $q_{2}^{(L)}$)  and the minimum value $-M_2^{(G)}$ (resp. $-M_2^{(L)}$) at $-q_{2}^{(G)}$ (resp. $-q_{2}^{(L)}$) . 

\noindent
{\it Remark: } 
Note that the function $2q v''(q^2)$ is a derivative of $v'(q^2)$ and the value $M_2$ 
attains the maximum value for the derivative of $v'(q^2)$. 

Here we assume 
\begin{eqnarray}
q_2 < q_f, 
\end{eqnarray}
and then introduce a function 
\begin{eqnarray}
L(q) :=  - 2q _f -2\k q \left( M_2 q  - 1\right)  -\k  \varepsilon  . 
\end{eqnarray}
Using the relation $v'(q^2)  < M_2 q +  v'(0) = M_2 q - 1$, which is derived by the fact that 
$M_2$ represents the maximal slope of the function $v'(q^2)$ for $0 < q < q_f$, 
leading to a trivial relation 

\begin{eqnarray}
\frac{v'(q^2) - v'(0)}{q- 0} < M_2. 
\end{eqnarray}
From the definition for the function $L(q)$, we have
\begin{eqnarray}
-2q _f -2\k q   v'(q^2) - \k \varepsilon > L(q). 
\end{eqnarray}
Combining this with the inequality (\ref{eq:upper_F(q)}),
we can easily show that the function $L(q)-q_f$ provides an lower bound of $F(q)-q_f$ (see Fig. \ref{fig:L(a)_F(q)}) : 
\begin{eqnarray}
\label{eq:F(q)>L(q)}
F(q) -q_f > L(q) -q_f .
\end{eqnarray}
Next, let  $\omega_1$ and $\omega_2$ be the solutions of 
$L(q)-q_f =q_f$ (see figure 15).
More explicitly, we have 
\begin{eqnarray} 
\label{eq:omega1}
\omega_1 = \frac{1}{2 M_2} - \sqrt{\Delta(\kappa, q_b, q_f)}, \\
\label{eq:omega2}
\omega_2 =  \frac{1}{2 M_2} +\sqrt{\Delta(\kappa, q_b, q_f)}, 
\end{eqnarray}
where 
\begin{eqnarray} 
\label{eq:discrimant}
\Delta(\kappa, q_b, q_f) := \frac{1}{4 M_2^2}  - \frac{\varepsilon}{2 M_2} - \frac{2 q_f }{M_2}\cdot \frac{1}{\kappa}. 
\end{eqnarray}
To make both $\omega_1$ and $\omega_2$ be real and positive, 
we have to require $\Delta(\kappa, q_b, q_f) >0$, 
which is written as 
\begin{eqnarray}
\label{eq:positive_discriminant_k}
\k >  \k_2,  
\end{eqnarray}
where 
\begin{eqnarray}
\label{def:kappa2}
\k_2 := \frac{ 4 q_f}{\displaystyle \frac{1}{2 M_2 }-  \varepsilon }. 
\end{eqnarray}
Here, to derive the above inequality, the condition 
\begin{eqnarray}
\label{eq:positive_discriminant}
\varepsilon < \varepsilon_2 
\end{eqnarray}
where
\begin{eqnarray}
\varepsilon_2 := \frac{1}{2 M_2 }, 
\end{eqnarray}
is assumed. We will examine the condition (\ref{eq:positive_discriminant}) in \ref{app:Gaussian_qf_qb} and \ref{app:Lorentzian_qf_qb}. 
\begin{figure}[H]
        \centering
	\includegraphics[width = 9.0cm,bb = 0 0 461 346]{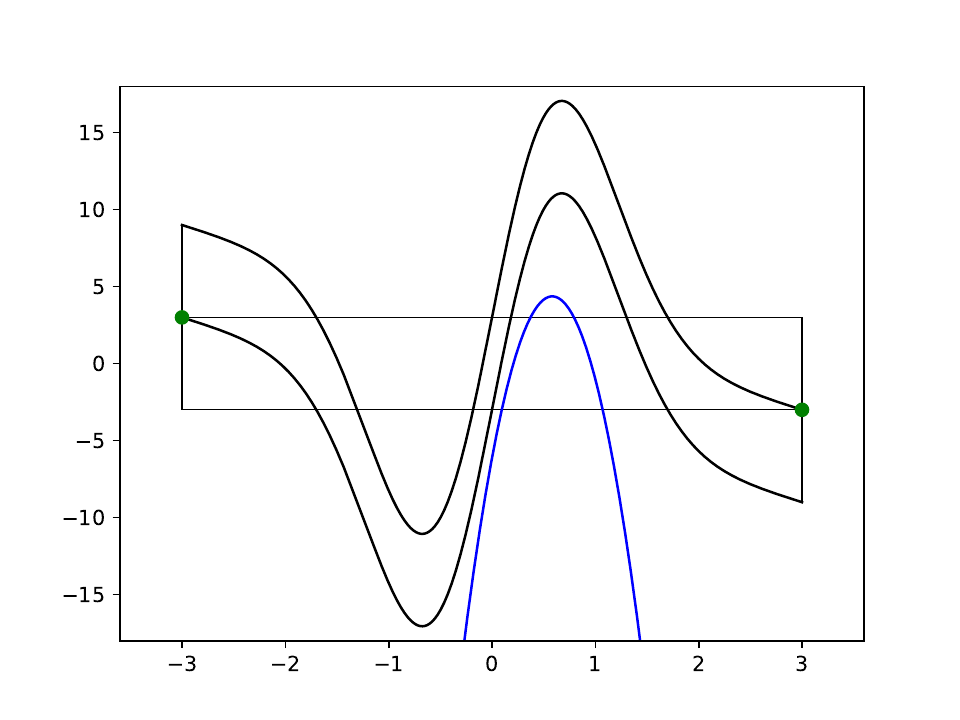}

	\caption{
	The boundary of the region $U^{-1}(R)$ (black) and the function $L(q)$ (blue). $L(q)$ is used as an lower bound for a boundary curve of the region $R$.
\label{fig:L(a)_F(q)}
	}
\end{figure} 

\begin{figure}[H]
        \centering
	\includegraphics[width = 9.0cm,bb = 0 0 461 346]{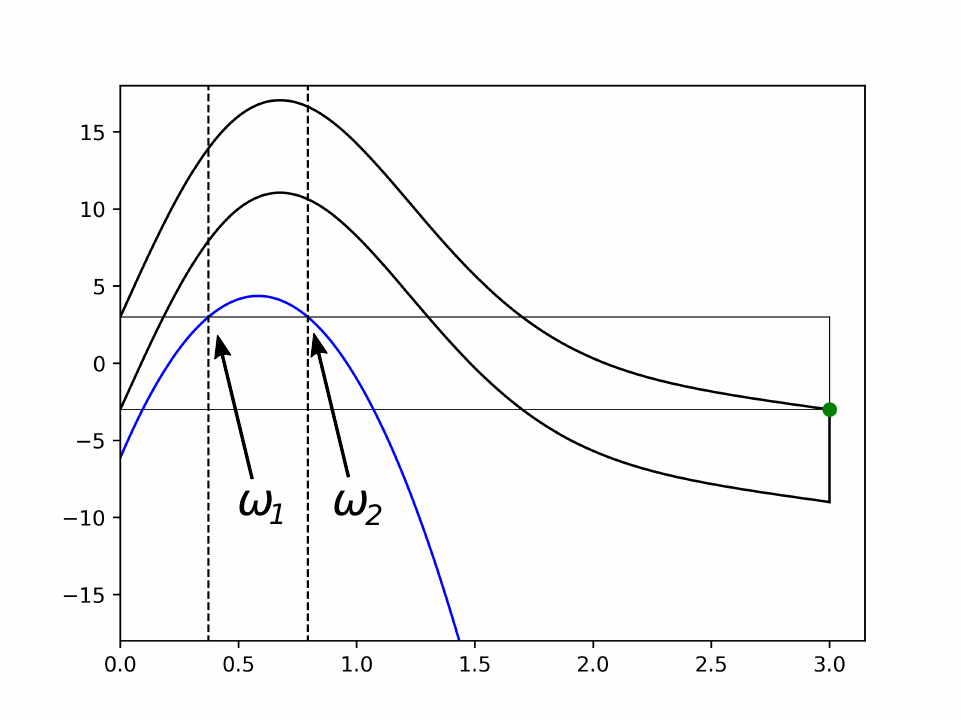}
	\caption{The zeros $\omega_1,\omega_2$ of the function $L(q)$ (blue). The black curves show the boundary of the region $U^{-1}(R)$. 
	\label{fig:omega}}
\end{figure} 
\subsection{Division of the phase space}
\label{sec:division_phase_space}

We next divide the phase space into three subregions: 
\begin{eqnarray}
\bar{\mathcal{X}} &=&\{ (q,p) \, |\, -q_f < q < - \omega_2\} ,\\
\bar{\mathcal{Y}} &=&\{ (q,p) \,| \,- \omega_1 < q < \omega_1\} ,\\
\bar{\mathcal{Z}} &=&\{ (q,p) \,| \, \omega_2 < q < q_f\} .
\end{eqnarray}
Figure \ref{fig:xyz} illustrates the subregions $\bar{\mathcal{X}}, \bar{\mathcal{Y}}$ and $\bar{\mathcal{Z}}$ 
introduced above. Recall that $\omega_1$ and $\omega_2$ are determined in such a way that 
the inequality (\ref{eq:F(q)>L(q)}) holds, so 
$R \cap U^{-1}(R) \subset  \bar{\mathcal{X}}\cup\bar{\mathcal{Y}}\cup \bar{\mathcal{Z}}$ holds, 
which implies $\Omega(U) \subset \bar{\mathcal{X}}\cup\bar{\mathcal{Y}}\cup \bar{\mathcal{Z}}$
from the Proposition \ref{prp:nonwandering}. 
If the points contained in the set $ \bar{\mathcal{X}}\cup\bar{\mathcal{Y}}\cup \bar{\mathcal{Z}}$ satisfy (\ref{sector_cond}), then the nonwandering set turns out to be uniformly hyperbolic. 
\begin{figure}[H]
\begin{center}
\includegraphics[width = 10.0cm,bb = 0 0 461 346]{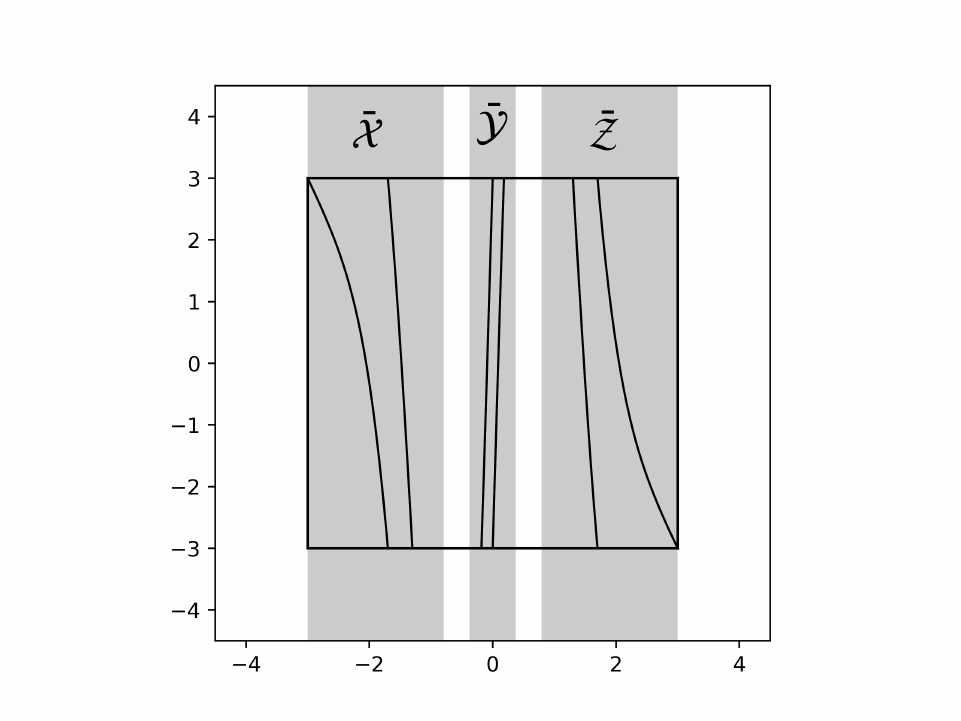}

\end{center}
\caption{
The subregions $\bar{\mathcal{X}}$, $\bar{\mathcal{Y}}$ and $\bar{\mathcal{Z}}$ in phase space. 
The borders are determined by the zeros $\omega_1,\omega_2$ of 
the equation $L(q) = 0$. 
} 
\label{fig:xyz}
\end{figure}
\subsection{Sufficient Conditions for Conley-Moser Condition}
\label{sec:Conditions}
Recalling an explicit form of $F'(q)$: 
\begin{eqnarray}
F'(q) = -2- 2\kappa (v'(q^2)+2q^2v''(q^2))\nonumber \\ 
\quad\quad\quad\quad \quad \quad -2\kappa \varepsilon \left( (q-q_b)v'((q-q_b)^2)-(q+q_b)v'((q+q_b)^2) \right), 
\end{eqnarray}
and taking into account the symmetry with respect to the $p$-axis, 
we rewrite a sufficient condition for the sector condition(\ref{sector_cond}) as
\begin{eqnarray}
\label{cond:Ybar}
 \quad F'(q) > \tilde{\mu}, \quad for \quad 0 < q < \omega_1,  \\
\label{cond:Zbar}
\quad F'(q) < -\tilde{\mu}, \quad for \quad \omega_2 < q < q_f. 
\end{eqnarray}
Then we find the following proposition: 
\begin{prp}
If the parameters $\varepsilon$ and $\kappa$ satisfy 
\begin{eqnarray}
\label{cond_prop62}
\varepsilon < \varepsilon_3, ~~~ \k > \k_3, 
\end{eqnarray}
where 
\begin{eqnarray}
\label{eq:epkp3}
\varepsilon_3 := - \frac{1}{2M_1} \left(v'\left(\left(\frac{1}{2M_2}\right)^2\right)+\frac{1}{2M_2^2}v''\left(\left(\frac{1}{2M_2}\right)^2\right)\right), 
\end{eqnarray}
and 
\begin{eqnarray}
\label{def:kappa3}
\k_3 := \frac{\displaystyle 2 + \tilde{\mu}}{\displaystyle - 2 \left(v'\left(\left(\frac{1}{2M_2}\right)^2\right)+\frac{1}{2M_2^2}v''\left(\left(\frac{1}{2M_2}\right)^2\right)\right) -4\varepsilon M_1}, 
\end{eqnarray}
then the following holds:
\begin{eqnarray}
F'(q) > \tilde{\mu}, \quad for \quad 0 < q < \omega_1. 
\end{eqnarray}
\end{prp}
\begin{proof}
We first introduce the function giving a lower bound of $F'(q)$:
\begin{eqnarray}
\beta_1(q) := -2- 2\kappa (v'(q^2)+2q^2v''(q^2)) -4\kappa \varepsilon M_1
\end{eqnarray}
Since the maximum and minimum values of the function $qv'(q^2)$ are 
respectively $\pm M_1$,  we have 
\begin{eqnarray}
F'(q) > \beta_1(q), 
\end{eqnarray}
and $\beta_1(q)$ monotonically decrease for $0<q<q_0$. 
Now using an inequality for $\omega_1$, 
\begin{eqnarray*}
\omega_1  = \frac{1}{2M_2} - \sqrt{ \frac{1}{4 M_2^2}  -\varepsilon \frac{1}{2 M_2} - \frac{2 q_f }{M_2}\cdot \frac{1}{\kappa}} < \frac{1}{2M_2}, 
\end{eqnarray*}
with concrete values for $M_2$ in each potential case, 
it is easy to verify 
\begin{eqnarray*}
\frac{1}{2M_2^{(G)}} =  \sqrt{\frac{e}{8}} < \sqrt{\frac{3}{2}} = q_0^{(G)} , \\
\frac{1}{2M_2^{(L)}} = \frac{27}{25\sqrt{5}} < 1 = q_0^{(L)},  
\end{eqnarray*}
which leads to the inequality $\omega_1 <1/2M_2< q_0$. 
Hence if $\beta_1(\omega_1) > \tilde{\mu}$ is satisfied, 
the inequality (\ref{cond:Ybar}) holds. 
In this situation, since the inequality $\omega_1 <1/2M_2< q_0$ leads to 
$\beta_1(\omega_1) > \beta_1(1/2M_2)$, the inequality (\ref{cond:Ybar}) holds 
provided that  $\beta_1(1/2M_2) > \tilde{\mu}$ is satisfied. 
The condition $\beta_1(1/2M_2) > \tilde{\mu}$ is explicitly written as
\begin{eqnarray}
-2- 2\kappa \left(v'\left(\left(\frac{1}{2M_2}\right)^2\right)+\frac{1}{2M_2^2}v''\left(\left(\frac{1}{2M_2}\right)^2\right)\right) -4\kappa \varepsilon M_1 > \tilde{\mu}. 
\end{eqnarray}
From this inequality, if 
\begin{eqnarray}
\label{eq:ep3}
\varepsilon < - \frac{1}{2M_1} \left(v'\left(\left(\frac{1}{2M_2}\right)^2\right)+\frac{1}{2M_2^2}v''\left(\left(\frac{1}{2M_2}\right)^2\right)\right)
\end{eqnarray}
is satisfied, then we can obtain the condition for $\kappa$ as
\begin{eqnarray}
\kappa > \frac{\displaystyle 2 + \tilde{\mu}}{\displaystyle - 2 \left(v'\left(\left(\frac{1}{2M_2}\right)^2\right)+\frac{1}{2M_2^2}v''\left(\left(\frac{1}{2M_2}\right)^2\right)\right) -4\varepsilon M_1}. 
\end{eqnarray}
\end{proof}
\noindent
{\it Remark.}
The conditions for $q_b$ and $q_f$ ensuring the inequality (\ref{eq:ep3})
 are examined in \ref{app:Gaussian_qf_qb} and \ref{app:Lorentzian_qf_qb}. 

Next we prepare the lemma which will be used to prove 
the subsequent proposition. 
\begin{lmm}
\label{lmm:q1_o2}
If the parameters $\varepsilon$ and $\kappa$ satisfy 
\begin{eqnarray}
\label{eq:ep4_epsilon}
\varepsilon < \varepsilon_4, ~~~ \k > \k_4, 
\end{eqnarray}
where 
\begin{eqnarray}
\label{eq:ep4}
\varepsilon_4 := -2M_2 {q_1}^2 + 2q_1, 
\end{eqnarray}
and
$\kappa$ is
\begin{eqnarray}
\label{cond:kappa_4}
\k_4 :=  \frac{\displaystyle 2 q_f }{\displaystyle -M_2 {q_1}^2 + q_1 -\frac{1}{2} \varepsilon}, 
\end{eqnarray}
then 
\begin{eqnarray}
q_1 < \omega_2
\end{eqnarray}
holds. 
\end{lmm}
\begin{proof}
The inequality 
\begin{eqnarray*}
\kappa > \frac{\displaystyle 2 q_f }{\displaystyle -M_2 {q_1}^2 + q_1 -\frac{1}{2} \varepsilon}
\end{eqnarray*}
is rewritten as 
\begin{eqnarray*}
\left( q_1 - \frac{1}{2M_2} \right)^2 < \frac{1}{4 M_2^2}  - \frac{\varepsilon}{2 M_2} - \frac{2 q_f }{M_2}\cdot \frac{1}{\kappa}. 
\end{eqnarray*}
This leads to 
\begin{eqnarray*}
q_1 - \frac{1}{2M_2}  < \sqrt{\frac{1}{4 M_2^2}  - \frac{\varepsilon}{2 M_2} - \frac{2 q_f }{M_2}\cdot \frac{1}{\kappa}}, 
\end{eqnarray*}
which is rewritten as 
\begin{eqnarray*}
q_1    < \frac{1}{2M_2} + \sqrt{\frac{1}{4 M_2^2}  - \frac{\varepsilon}{2 M_2} - \frac{2 q_f }{M_2}\cdot \frac{1}{\kappa}}.  
\end{eqnarray*}
Recall the explicit expression for the discriminant $\Delta(\kappa, q_b, q_f)$ 
given in (\ref{eq:discrimant}), we find that 
\begin{eqnarray*}
q_1 < \omega_2. 
\end{eqnarray*}
\end{proof}

Next we provide a sufficient condition leading to 
(\ref{cond:Zbar}). 
\begin{prp}
If the condition
\begin{eqnarray}
q_f < q_b - q_0 
\end{eqnarray}
is satisfied respectively  for the Gaussian and Lorentzian potential case 
where 
\begin{eqnarray}
q_{0}^{(G)} = \sqrt{\frac{3}{2}},\\
q_{0}^{(L)} = 1. 
\end{eqnarray} 
Moreover, if the parameters $\varepsilon$ and $\kappa$ satisfy 
\begin{eqnarray}
\label{eq:e5_epsilon}
\varepsilon < \varepsilon_5, ~~~\k > \max\{\k_5,\k_6 \} , 
\end{eqnarray}
where 
\begin{eqnarray}
\varepsilon_5 := \frac{2}{3M_2} -\frac{2b}{3a}M_2, 
\end{eqnarray}
and 
\begin{eqnarray}
\label{def:kappa_5}
\k_5 := \frac{\displaystyle \tilde{\mu}-2 + \frac{12aq_f}{M_2} v''(1/M_2^2)}{\displaystyle \left( \frac{2a}{M_2^2} -2b -\frac{3a}{M_2}\varepsilon \right)v''(1/M_2^2)},\\
\label{def:kappa_6}
\k_6 := \frac{\displaystyle \tilde{\mu} -2 }{\displaystyle 2(aq_f^2 - b)v''(q_f^2)}, 
\end{eqnarray}
then we have 
\begin{eqnarray}
F'(q) < -\tilde{\mu}, \quad for \quad \omega_2 < q < q_f. 
\end{eqnarray}
Here the constants $a$ and $b$ are given by 
\begin{eqnarray}
a^{(G)} = 2,\quad b^{(G)} = 1,  \\
a^{(L)} = \frac{3}{2},\quad b^{(L)} = \frac{1}{2}, 
\end{eqnarray}
for the Gaussian and Lorentzian potential case, respectively. 
\end{prp}
\begin{proof}
We first introduce the function giving an upper bound of $F'(q)$:
\begin{eqnarray}
\beta_2(q) := -2- 2\kappa (v'(q^2)+2q^2v''(q^2)) . 
\end{eqnarray}
As is the case in the Lemma \ref{lmm:extrema}, 
we here assume the same inequality as (\ref{con_mild}). 
Since for $0 < q < q_f$, 
\begin{eqnarray}
(q-q_b)v'((q-q_b)^2)-(q+q_b)v'((q+q_b)^2) > 0
\end{eqnarray}
is satisfied, it is easy to show that 
\begin{eqnarray}
F'(q) < \beta_2(q)
\end{eqnarray}
holds. 
Taking into account of the fact that $1/2M_2 < \omega_2 < 1/M_2$, 
and using concrete values for $q_0$ in each potential case 
(see Definition \ref{dfn:q0}),  we obtain 
\begin{eqnarray}
\frac{1}{M_2^{(G)}} =  \sqrt{\frac{e}{2}} < \sqrt{\frac{3}{2}} = q_0^{(G)} , \\
\frac{1}{M_2^{(L)}} = \frac{54}{25\sqrt{5}} < 1 = q_0^{(L)}. 
\end{eqnarray}
Hence we can show that $\beta_2(q)$ monotonically decreases for $\omega_2<q<q_0$, 
and monotonically increases for $q_0<q<q_f$. 
Therefore, if $\beta_2(\omega_2)<-\tilde{\mu}$ and $\beta_2(q_f)<-\tilde{\mu}$ 
are satisfied, then the condition (\ref{cond:Zbar}) follows. 

We first provide the condition for $\k$ derived from 
the condition $\beta_2(\omega_2)<-\tilde{\mu}$. 
To this end, we rewrite the function $\beta_2$ as
\begin{eqnarray}
\beta_2(\omega_2) = -2- 2\kappa \left(\frac{v'(\omega_2^2)}{v''(\omega_2^2)}+2\omega_2^2 \right)v''(\omega_2^2). 
\end{eqnarray}
Note that an explicit expression in the large bracket takes the form as
\begin{eqnarray}
\frac{{v'}(q^2)}{{v''}(q^2)}+2q^2 = aq^2 - b, 
\end{eqnarray}
where 
\begin{eqnarray}
a^{(G)} = 2,\quad b^{(G)} = 1,  \\
a^{(L)} = \frac{3}{2},\quad b^{(L)} = \frac{1}{2}. 
\end{eqnarray}
Note that $v''(q^2)$ monotonically decreases for $q>0$ and 
Lemma \ref{lmm:q1_o2} leads to the inequality 
 $a \omega_2^2 - b > 0$. 
The latter is deduced from the concrete expression for $aq^2 - b$. 
As a result, 
$-  \left(a\omega_2^2 -b \right) v''(q^2)$ turns out to be a monotonically 
increasing function for $q>0$.  
Recall the trivial inequality $1/(2 M_2) < \omega_2 < 1/M_2$, which immediately leads to 
\begin{eqnarray}
v''(\omega_2^2 )> v''(1/M_2^2 ).
\label{ineq:v''}
\end{eqnarray}
Using (\ref{ineq:v''}), we have 
\begin{eqnarray}
\beta_2(\omega_2) < -2- 2\kappa \left(a\omega_2^2 -b \right) v''(1/M_2^2 ). 
\end{eqnarray}
Recall the form (\ref{eq:discrimant}) of the discriminant $\Delta(\kappa, q_b, q_f)$,
and 
a trivial inequality $\varepsilon /(2 M_2) + 2 q_f  / (M_2 \kappa) > 0$, we can show that 
\begin{eqnarray}
\Delta(\kappa, q_b, q_f) <  \frac{1}{{2M_2}} \sqrt{\Delta(\kappa, q_b, q_f)} . 
\end{eqnarray}
and
\begin{eqnarray}
\omega_2^2 > \frac{1}{4M_2^2}+ 3\Delta(\kappa, q_b, q_f).
\end{eqnarray}
From this fact, we find 
\begin{eqnarray}
\beta_2(\omega_2) < -2- 2\kappa \left(a\left( \frac{1}{4M_2^2}+ 3\Delta \right) -b \right) v''(1/M_2^2 ).
\end{eqnarray}
Hence the inequality 
\begin{eqnarray}
-2- 2\kappa \left(a\left( \frac{1}{4M_2^2}+ 3\Delta \right) -b \right)  v''(1/M_2^2 )
< -\tilde{\mu}
\end{eqnarray}
is a sufficient condition leading to $\beta_2(\omega_2)<-\tilde{\mu}$. 
Inserting the form of the discriminant (\ref{eq:discrimant}), we have 
\begin{eqnarray}
\label{eq:kp5_1_1}
-2- 2\kappa \left(a\left( \frac{1}{M_2^2} - \frac{3}{2M_2}\varepsilon - \frac{6q_f}{M_2}\cdot \frac{1}{\k} \right) -b \right) v''(1/M_2^2 )
  < -\tilde{\mu}, \nonumber 
\end{eqnarray}
which is rewritten as 
\begin{eqnarray}
\label{eq:kp5_1_2}
\k\left( \left( \frac{2a}{M_2^2} -2b -\frac{3a}{M_2}\varepsilon \right) v''(1/M_2^2 ) 
 \right) > \tilde{\mu}-2 + \frac{12aq_f}{M_2} v''(1/M_2^2 ).
\end{eqnarray}
Assuming that 
\begin{eqnarray}
\label{eq:ep5_1_3}
\frac{2a}{M_2^2} -2b -\frac{3a}{M_2}\varepsilon > 0 , 
\end{eqnarray}
or equivalently, 
\begin{eqnarray}
\label{eq:ep5}
\varepsilon < \frac{2}{3M_2} -\frac{2b}{3a}M_2, 
\end{eqnarray}
we finally reach the condition for $\k$ as
\begin{eqnarray}
\label{eq:kp5_1_4}
\kappa > \frac{\displaystyle \tilde{\mu}-2 + \frac{12aq_f}{M_2} v''(1/M_2^2 )}{\displaystyle \left( \frac{2a}{M_2^2} -2b -\frac{3a}{M_2}\varepsilon \right)v''(1/M_2^2 )}. 
\end{eqnarray}

Our next task is to provide the condition derived from the second 
condition $\beta_2(q_f)<-\tilde{\mu}$, 
but this can easily be achieved
by solving 
this explicitly for $\k$ to give
\begin{eqnarray}
\kappa > \frac{\displaystyle \tilde{\mu} -2 }{\displaystyle 2(aq_f^2 - b)v''(q_f^2)}. 
\end{eqnarray}
\end{proof}
\noindent
{\it Remark.}
The conditions for $q_b$ and $q_f$ ensuring the first inequality in (\ref{eq:ep4_epsilon})
 are examined in \ref{app:Gaussian_qf_qb} and \ref{app:Lorentzian_qf_qb}.

\section{Uniform hyperbolicity theorems}

Collecting all the conditions obtained so far, we finally provide 
sufficient conditions for Gaussian and Lorentzian potential cases. 
\begin{mth}[Gaussian potential]
Let $q_b$ and $q_f$ be taken such that the following conditions
\begin{eqnarray}
q_f > \frac{3}{\sqrt{2}}
\end{eqnarray}
and 
\begin{eqnarray}
q_f + \sqrt{\frac{3}{2}}<q_b < q_f + \sqrt{q_f^2 -\log{q_f} + \log \left(\varepsilon^{(G)}_{3} \frac{e^2-1}{2 e^2}\right)}, 
 \end{eqnarray}
are satisfied. 
$\varepsilon^{(G)}_{3}$(\ref{eq:epkp3}) is the constant given as
\begin{equation}
 \varepsilon^{(G)}_{3} = \sqrt{\frac{e}{2}}\left( 1- \frac{e}{4}\right)e^{-\frac{e}{8}}\approx 0.26594
\end{equation}
Recall $\kappa_i(1\leq i \leq 6)$ (\ref{def:kappa1}),(\ref{def:kappa2}),(\ref{def:kappa3}),(\ref{cond:kappa_4}),(\ref{def:kappa_5}), and (\ref{def:kappa_6}), respectively.
Then, for 
\begin{eqnarray}
\label{kappa_cond_G}
\k^{(G)} > \max_{1\leq i \leq 6}\kappa^{(G)}_i
\end{eqnarray}
the nonwandering set of the map (\ref{eq:map1}) with the Gaussian potential (\ref{eq:Gaussian}) 
is a topological horseshoe and uniformly hyperbolic. 
\end{mth}
\begin{mth}[Lorentzian Potential]
Let $q_b$ and $q_f$ be taken such that the following conditions
\begin{eqnarray}
q_b >4
\end{eqnarray}
and 
\begin{eqnarray}
\sqrt{ \frac{\displaystyle -q_b^2 +q_b\sqrt{2\varepsilon_3q_b^3 + 8\varepsilon_3q_b -4}}{2\varepsilon_3 q_b-1} - 1}< q_f <\sqrt{\frac{1}{2}q_b^2 + 1} 
\end{eqnarray}
are satisfied. 
$\varepsilon^{(L)}_{3}$ (\ref{eq:epkp3}) is the constant given as
\begin{equation}
\varepsilon^{(L)}_{3} = \frac{8}{3\sqrt{3}} \frac{5^{10}}{\displaystyle \left( 5^5 + 3^6\right)^3}\left( 5^5 - 3^7\right)
 \approx 0.24636
\end{equation}
Recall $\kappa_i(1\leq i \leq 6)$ (\ref{def:kappa1}),(\ref{def:kappa2}),(\ref{def:kappa3}),(\ref{cond:kappa_4}),(\ref{def:kappa_5}), and (\ref{def:kappa_6}), respectively.
Then, for 
\begin{eqnarray}
\label{kappa_cond_L}
\k^{(L)} > \max_{1\leq i \leq 6}\kappa^{(L)}_i
\end{eqnarray}
the nonwandering set of the map (\ref{eq:map1}) with the Lorentzian potential (\ref{eq:Lorentzian}) 
is a topological horseshoe and uniformly hyperbolic. 
\end{mth}

In the Appendices A and B, we examine which constant $\k_i (1 \le i \le 6)$ provides 
the strongest condition in (\ref{kappa_cond_G}) and (\ref{kappa_cond_L}). 
It turns out that for the Gaussian potential case this is 
the condition associated with the behavior of the potential function around 
the fixed points while the conditions coming from another region become most dominant for the Lorentzian potential case. 
Note that this difference originates from the fact that stable and unstable manifolds for 
fixed point emanate with a finite but small angle in the Gaussian potential case, while
the corresponding angle for the Lorentzian case is larger. 

\section{Summary and discussion}
In recent years fractal Weyl laws and related eigenfunction hypothesis have numerically been studied in a wealth of model systems. These comprise of baker maps with artificial projective openings, billiard systems, and artificial mathematical paradigms such as the H{\'e}non map or systems of constant negative curvature.
In this work, we focus on simple 2-dimensional scattering maps, which admit a closed-form expression.
For these 2-dimensional scattering maps, previously no model system existed, which unifies all of the following desirable properties at once:
(i) A proof of uniform hyperbolicity in some parameter region.
(ii) Being amenable to a well-established method for numerically treating quantum resonance states.
(iii) Obeying free motion in the asymptotic scattering regions.
And (iv) comprising of fully analytic potentials.
To be specific, most of the Baker maps being used, utilize artificial projective openings and hence violate property (iii) and (iv). Model systems such as the H{\'e}non map or systems of constant negative curvature on the other hand, do not obey free motion in the asymptotic region. Hence, they violate condition (iii).

In this article, we establish the first family of 2-dimensional scattering systems which obey all four desirable properties. More specifically, our previous work \cite{NS18} demonstrated properties (ii), (iii), and (iv) and in particular the correct numerical treatment via complex scaling and weak absorbing potentials. In contrast, this paper adds the proof of uniform hyperbolicity for the proposed scattering systems in the limit of large kick-strengths and thus contributes property (i). We believe that this should be valuable to further test fractal Weyl laws and related topics such as chaotic eigenfunction hypothesis \cite{Ketz18}. To illustrate this in more detail, we now discuss a couple of follow on studies.

A first follow on study from this paper would need to establish that the fractal properties of the repeller in our model system can be controlled to some degree. While analytical methods may be out of question it is conceivable that fractal properties of the classical repeller and trapped set can be tuned by successively lowering the kick strength $\kappa$. Establishing such properties could be done using the usual numerical methods of box counting \cite{Lin02,Ketz13,Novaes13,Altmann}. Alternatively, one could employ rigorous mathematical methods using computer assisted methods \cite{Arai08}.

Next, after establishing that fractal properties can be controlled via the kick-strength $\kappa$, it would be valuable to visit original numerical studies of the fractal Weyl law \cite{Lin02,Zworski02,Lu03,Schomerus,Nonnen05,Ketz13,Borthwick14}.
In particular, using the absorbing potential method described in Ref.~\cite{NS18} it should be possible to tune our model system from the regime where quantum resonances are treated with an exact numerical method that produces quantum resonances identical to complex scaling (weak absorbing potentials), and continuously cross over to the regime of projective openings used by Ref. \cite{Nonnen05}.
This would allow to assess whether their original observations still hold and whether there are any additional effects to be observed, for example due to diffractive effects from the absorbing potentials.

Finally, it would further be valuable to validate that results related to eigenfunction rather than eigenvalue hypothesis in chaotic scattering systems \cite{Ketz18}.
In particular, it would be valuable to validate whether such results remain valid, if considered in more realistic model systems as ours where quantum resonance states can be computed with theoretically grounded methods.

\ack
The authors are grateful to Zin Arai for his valuable comments. 
This work has been supported by JSPS KAKENHI Grant
Numbers 17K05583. 
N.M. acknowledges financial support by Deutsche Forschungsgemeinschaft(DFG) via Grant No. ME 4587/1-1, in the very early stage of this project.
\appendix

\section{Conditions for the Gaussian function}
\label
{app:Gaussian}
\subsection{Some properties of $V'(q)$}
\label{app:Gaussian_out}

In this Appendix, 
we show that the derivative $V'(q)$ has zeros only at $q = \pm q_f$ and $q=0$, 
and that $V'(q) < 0$ holds for $q > q_f$. 

As for the zeros of $V'(q)$, it is enough to show that 
the equation $V'(q) = 0$ has a solution $q = q_f$ for $q > 0$ because $V'(-q) =-V'(q)$ 
and $V'(0)=0$ are satisfied. 
From an explicit form of $V'(q)$: 
\begin{eqnarray}
\label{app:derivative}
V'(q) = 2 \kappa q e^{-q^2} \left( q - q_f \frac{\sinh{2q_b q}}{\sinh{2q_b q_f}} \right), 
\end{eqnarray}
we can easily see that $q=q_f$ is a solution of $V'(q) =0$. 
The first derivative of the terms in the bracket give 
\begin{eqnarray}
1 - 2q_b q_f \frac{\cosh{2q_b q}}{\sinh{2q_b q_f}} 
\end{eqnarray}
and the second derivative, 
\begin{eqnarray}
- 2{q_b}^2 q_f \frac{\sinh{2q_b q}}{\sinh{2q_b q_f}}. 
\end{eqnarray}
The latter is negative for $q>0$ because $q_b$ and $q_f$ are positive constants, 
thus the first derivative monotonically decreases for $q>0$. 
Furthermore, the first derivative takes the value of
\begin{eqnarray}
1 -  \frac{2q_b q_f}{\sinh{2q_b q_f}} 
\end{eqnarray}
at $q=0$, which turns to be positive using the fact that 
that the inequality $x/\sinh x < 1$ holds for $x > 0$. 
On the other hand, at $q = q_f$, the first derivative is given as 
\begin{eqnarray}
1 -  \frac{2q_b q_f}{\tanh{2q_b q_f}} , 
\end{eqnarray}
which is negative because since the inequality $x/\tanh x > 1$ holds for $x > 0$. 
Consequently, we can say that the terms in the bracket in (\ref{app:derivative})
 has a unique extremum, and so that there are no zeros in the range of $0 < q < q_f$
 since  $q=0$ and $q_f$ are zeros of the function $V'(q)$. 
The above consideration also implies that the terms of the bracket 
is negative for $q_f < q$, which results in the property, 
\begin{eqnarray}
V'(q) < 0 \quad {\rm for} \quad q_f < q. 
\end{eqnarray}

\subsection{The proof of Lemma \ref{lmm:extrema}}
\label{app:Gaussian_mild}
We here prove that 
$F(q) -q_f < q_f$ or equivalently $F(q) -2q_f < 0$ for $-q_f<q<0$, 
and $F(q)$ has a unique local maximum 
in $0<q<q_f$ for the Gaussian potential case. 

First we rewrite $F^{(G)}(q) - 2q_f$ as 
\begin{eqnarray}
F^{(G)}(q) -2q_f
&=& -2q -2q_f + 2\kappa q e^{-q^2} - \k \varepsilon \left( e^{-(q-q_b)^2} -e^{-(q+q_b)^2} \right) \nonumber \\
&=& 
 -2q -2q_f + 2 \k  e^{- q^2} \widetilde{F^{(G)}}(q), 
\end{eqnarray}
where 
\begin{eqnarray}
\widetilde{F^{(G)}}(q) := q - \frac{1}{2} \varepsilon e^{-{q_b}^2}  \left( e^{2 q_b q} -e^{-2q_b  q} \right). 
\end{eqnarray}
For the Gaussian case, we have 
\begin{eqnarray}
\varepsilon = \frac{q_f}{e^{-q_b^2} \sinh(2q_fq_b)}. 
\end{eqnarray}

Since 
\begin{eqnarray}
\widetilde{F^{(G)}}'(q) &=& 1 - \varepsilon q_be^{-{q_b}^2}  \left( e^{2q_b q} +e^{-2q_b  q} \right) \nonumber \\
 &=& 1 - \frac{q_f q_b}{\sinh{2q_f q_b}} \left( e^{2q_b q} +e^{-2q_b  q} \right), 
\end{eqnarray}
we can find the solutions for $\widetilde{F^{(G)}}'(q) = 0$ as
\begin{eqnarray}
q = \pm \frac{1}{2q_b} \log {\frac{\sinh{2q_f q_b}+ \sqrt{\sinh^2{2q_f q_b}- 4q_f^2 q_b^2 }}{2q_f q_b}}. 
\end{eqnarray}
This implies that $\widetilde{F^{(G)}}(q)$  has a unique local minimum for $-q_f < q<0$. 
Using the mean value theorem, together with the fact $\widetilde{F^{(G)}}(-q_f) = \widetilde{F^{(G)}}(0) = 0$, 
we can say that $\widetilde{F^{(G)}}(q) < 0$ for $-q_f < q<0$. 
As a result, $ 2 \k  e^{-q^2} \widetilde{F^{(G)}}(q) < 0$ holds 
and thus $F(q) - 2q_f < 0$ follows for $-q_f < q<0$. 

Next we will show that 
${F^{(G)}}(q)$ has a unique local maximum for $0 < q < q_f$ provided that 
the condition 
\begin{eqnarray}
\label{app:G_assum}
q_f <q_b - \sqrt{\frac{3}{2}}
\end{eqnarray}
is satisfied (see figure A1). 
To this end we consider the behavior of the first and second derivative 
\begin{eqnarray}
\fl
{F^{(G)}}'(q) &=& -2  + 2 \kappa  (1 - 2 q^2) e^{-q^2} + 2 \k \varepsilon \left( \left(q-q_b\right)e^{-(q-q_b)^2} -\left(q+q_b\right)e^{-(q+q_b)^2} \right)  
\\ 
\fl
{F^{(G)}}''(q) &=& 4 \kappa q (2 q^2 - 3) e^{-q^2} \nonumber \\
\fl
&+& 2 \k \varepsilon \left( \left(1 -  2(q-q_b)^2\right)e^{-(q-q_b)^2} -\left(1 -  2(q+q_b)^2\right)e^{-(q+q_b)^2} \right) . 
\end{eqnarray}

We first consider the second derivative ${F^{(G)}}''(q)$. 
Due to the condition (\ref{app:G_assum}), 
the term $\left(1 -  2(q-q_b)^2\right)e^{-(q-q_b)^2}$ monotonically 
decreases for $0 < q < q_f$. 
In a similar way, the term $-\left(1 -  2(q+q_b)^2\right)e^{-(q+q_b)^2}$
monotonically decreases for $0 < q$.
Combining these with the fact that $\k>0$ and $\varepsilon>0$, 
\begin{eqnarray}
\label{app:G_kpep}
\fl
\quad \quad
2 \k \varepsilon \left( \left(1 -  2(q-q_b)^2\right)e^{-(q-q_b)^2} -\left(1 -  2( q+q_b)^2\right)e^{-(q+q_b)^2} \right)
\end{eqnarray}
monotonically decreases for $0 < q < q_f$.
In addition, since (\ref{app:G_kpep}) is zero at $q=0$,
\begin{eqnarray}
\label{app:G_hill}
\fl
\quad \quad
2 \k \varepsilon \left( \left(1 -  2(q-q_b)^2\right)e^{-(q-q_b)^2} -\left(1 -  2( q+q_b)^2\right)e^{-(q+q_b)^2} \right) < 0
\end{eqnarray}
holds for $0 < q < q_f$.
As for the second derivative ${F^{(G)}}''(q)$, 
we can show from the inequality (\ref{app:G_hill}) that 
\begin{eqnarray}
{F^{(G)}}''(q) <  4 \kappa q (2q^2 - 3) e^{-q^2}. 
\end{eqnarray}
Since $4 \kappa q (2 q^2 - 3) e^{-q^2} < 0 $ for $0 < q < \sqrt{3/2}$, 
the following holds for $0 < q < \sqrt{3/2}$
\begin{eqnarray}
\quad\quad {F^{(G)}}''(q) < 0. 
\end{eqnarray}
This shows that ${F^{(G)}}'(q)$ monotonically decreases 
for $0 < q < \sqrt{3/2}$. 
From (\ref{app:G_hill}), we can show that the terms 
\begin{equation}
2 \k \varepsilon \left( (q-q_b)e^{-(q-q_b)^2} -(q+q_b)e^{-(q+q_b)^2} \right)
\end{equation}
monotonically decreases for $0<q<q_f$, and take a negative value at $q=0$, 
Thus we have 
\begin{equation}
2 \k \varepsilon \left( (q-q_b)e^{-(q-q_b)^2} -(q+q_b)e^{-(q+q_b)^2} \right) < 0, 
\end{equation}
which leads to 
\begin{eqnarray}
\quad {F^{(G)}}'(q) < 2\kappa (1 - 2 q^2) e^{-q^2},
\end{eqnarray}
for $0 < q < q_f$. 
Now, the fact that $2 \kappa (1 - 2 q^2) e^{-q^2} < 0$ holds for $1/\sqrt{2}< q$ 
leads to 
the following inequality for $1/\sqrt{2}< q$, 
\begin{eqnarray}
\label{app:G_positive}
\quad {F^{(G)}}'(q) < 0. 
\end{eqnarray}
Hence the function $F^{(G)}(q)$ has at most a single 
local maximum for $0<q<q_f$.

\subsection{Conditions for $q_b$ and $q_f$}
\label{app:Gaussian_qf_qb}

We here check the conditions required for $q_b$ and $q_f$. 
First we recall the condition (\ref{con_mild}), which is 
necessary to prove the Lemma \ref{lmm:extrema} (see \ref{app:Gaussian_mild}). 
Next we take a look at the inequalities required to derive the conditions for $\varepsilon^{(G)}_i  (i = 1, \cdots, 5)$ (see (\ref{eq:horseshoe_condition3}),(\ref{eq:positive_discriminant}),(\ref{cond_prop62}),(\ref{eq:ep4_epsilon}) and (\ref{eq:e5_epsilon})) : 
\begin{eqnarray}
&& 
\varepsilon< \varepsilon^{(G)}_{1}  = \sqrt{\frac{2}{e}} \approx 0.85776, 
\\
&&
\varepsilon < \varepsilon^{(G)}_{2} = \sqrt{\frac{e}{8}} \approx 0.58291, 
\\
&&  \varepsilon< \varepsilon^{(G)}_{3} = \sqrt{\frac{e}{2}}\left( 1- \frac{e}{4}\right)e^{-\frac{e}{8}}\approx 0.26594,
 \\
&&
\varepsilon < \varepsilon^{(G)}_{4} =-\sqrt{\frac{2}{e}}+\sqrt{2} \approx 0.55644, 
\\
&&
\varepsilon < \varepsilon^{(G)}_{5} = \frac{1}{3}\sqrt{\frac{2}{e}} \left( e - 1\right)\approx 0.49129.
\end{eqnarray}
By comparing concrete values, we obtain 
\begin{eqnarray}
\varepsilon^{(G)}_{1} > 
\varepsilon^{(G)}_{2} > 
\varepsilon^{(G)}_{4} > 
\varepsilon^{(G)}_{5} > 
\varepsilon^{(G)}_{3}
\end{eqnarray}
and so finally conclude that the condition 
\begin{eqnarray}
\label{eq:Appendix_conditionC}
\varepsilon(q_b,q_f) < \varepsilon^{(G)}_{3}
\end{eqnarray}
turns out to be the strongest one. 
Now using an explicit expression (\ref{eq:epsilon}) for $\varepsilon$ for 
the Gaussian potential case, 
\begin{eqnarray}
\varepsilon &=&  -\frac{2q_f{v^{(G)}}'(q_f^2)}{v^{(G)} \bigl( (q_f-q_b)^2) - v^{(G)} ((q_f+q_b)^2 \bigr)  } \nonumber \\
&=&\frac{\displaystyle 2q_fe^{-q_f^2}}{\displaystyle e^{-(q_f-q_b)^2} - e^{-(q_f+q_b)^2}} , 
\end{eqnarray}
we can rewrite the condition (\ref{eq:Appendix_conditionC}) as 
\begin{eqnarray}
 \frac{\displaystyle 2q_fe^{-q_f^2}}{\displaystyle e^{-(q_f-q_b)^2} - e^{-(q_f+q_b)^2}}  < \varepsilon^{(G)}_3. 
\end{eqnarray}
Now if we assume that 
\begin{eqnarray}
\label{eq:suf_out_g}
2q_fq_b>1, 
\end{eqnarray}
then 
the inequality 
\begin{eqnarray}
\label{eq:ep_ineq}
e^{-(q_f-q_b)^2} - e^{-(q_f+q_b)^2} &=& e^{-(q_f-q_b)^2} (1 - e^{-4q_f q_b}) \nonumber \\
&>& e^{-(q_f-q_b)^2}(1 -e^{-2})
\end{eqnarray}
holds.
Using (\ref{eq:ep_ineq}), we can show that the condition 
\begin{eqnarray}
\label{eq:Appendix_conditionE}
 \frac{\displaystyle 2 q_f e^{-q_f^2}}{\displaystyle e^{-(q_f-q_b)^2}(1 -e^{-2})}  < \varepsilon^{(G)}_3 , 
\end{eqnarray}
leads to the condition (\ref{eq:Appendix_conditionC}) (see figure A2). 
Rewriting this to the inequality for $q_b$, we obtain
\begin{eqnarray}
\label{eq:q_f}
\fl
q_f - \sqrt{q_f^2 - \log q_f + \log \varepsilon_3\frac{e^2 -1 }{2e^2}} < q_b <  q_f + \sqrt{q_f^2 - \log q_f + \log \varepsilon_3\frac{e^2 -1 }{2e^2}} . 
\end{eqnarray}
Here note that 
\begin{eqnarray}
q_f - \sqrt{q_f^2 - \log q_f + \log \varepsilon_3\frac{e^2 -1 }{2e^2}} < q_f + \sqrt{\frac{3}{2}}. 
\end{eqnarray}
Combined this with the condition (\ref{con_mild}), 
it is enough to consider the following inequality
\begin{eqnarray}
\label{eq:A37}
q_f + \sqrt{\frac{3}{2}} < q_b <  q_f + \sqrt{q_f^2 - \log q_f + \log \varepsilon_3\frac{e^2 -1 }{2e^2}}. 
\end{eqnarray}
Differentiating the function in the root with respect to $q_f$, we find
\begin{eqnarray}
2q_f - \frac{1}{q_f} =\frac{1}{q_f}\left( 2 q_f^2 - 1 \right). 
\end{eqnarray}
Therefore, one of the branch that makes the function in the root positive should be 
contained in $ 0 < q_f < 1/ \sqrt{2}$, and the other in $q_f > 1/ \sqrt{2}$. 

First, we show it unnecessary to consider the case with $ 0 < q_f < 1/ \sqrt{2}$. 
Using the explicit value of $\varepsilon_3$, 
we obtain 
\begin{eqnarray}
\varepsilon_3\frac{e^2 -1 }{2e^2} < \frac{1}{\sqrt{2e}}, 
\end{eqnarray}
which allows us to rewrite the inequality (\ref{eq:A37}) as
\begin{eqnarray}
q_f + \sqrt{\frac{3}{2}} < q_b <  
q_f + \sqrt{q_f^2 - \log q_f + \log \frac{1}{\sqrt{2e}}} . 
\end{eqnarray}
On the other hand, the following relation 
\begin{eqnarray}
q_f + \sqrt{q_f^2 - \log q_f + \log \frac{1}{\sqrt{2e}}} 
< 
\frac{1}{2q_f}, 
\end{eqnarray}
holds because the function 
\begin{eqnarray}
 \frac{1}{2q_f} -  q_f - \sqrt{q_f^2 - \log q_f + \log \frac{1}{\sqrt{2e}}} \, .
\end{eqnarray}
monotonically decreases for $0 < q_f < 1/\sqrt{2} $ and is equal to zero at $q_f = 1/\sqrt{2}$. 
However, the inequality $q_b < 1/2q_f$ contradicts the assumption 
(\ref{eq:suf_out_g}). Hence, we do not have to consider the branch contained in $ 0 < q_f < 1/ \sqrt{2}$. 

As for the case with $q_f > 1/\sqrt{2}$, it is not clear whether there actually exists 
the domain $(q_b,q_f)$ in which the inequality (\ref{eq:A37}) holds. 
Here, we consider a sufficient condition using a concrete value of $q_f$, which 
guarantees that the inequality (\ref{eq:A37}) makes sense. 
To this end, notice that the function 
\begin{eqnarray}
\sqrt{q_f^2 - \log q_f + \log \varepsilon_3\frac{e^2 -1 }{2e^2}} -\sqrt{\frac{3}{2}}
\end{eqnarray}
monotonically increases for $q_f > 1/\sqrt{2}$, and it takes a positive value 
($\approx 0.03420\cdots$). 
Therefore, it is sufficient to impose the following conditions, 
\begin{eqnarray}
q_f > \frac{3}{\sqrt{2}}, 
\end{eqnarray}
and
\begin{eqnarray}
q_f + \sqrt{\frac{3}{2}} < q_b <  q_f + \sqrt{q_f^2 - \log q_f + \log \varepsilon_3\frac{e^2 -1 }{2e^2}}. 
\end{eqnarray}
It should be noted that the assumption (\ref{eq:suf_out_g}) is no more necessary 
since 
\begin{eqnarray}
q_f + \sqrt{\frac{3}{2}} > \frac{1}{2q_f}
\end{eqnarray}
holds for $q_f > 1/\sqrt{2}$. 
\begin{figure}[H]
\centering
\includegraphics[width = 10.0cm,bb = 0 0 461 346]{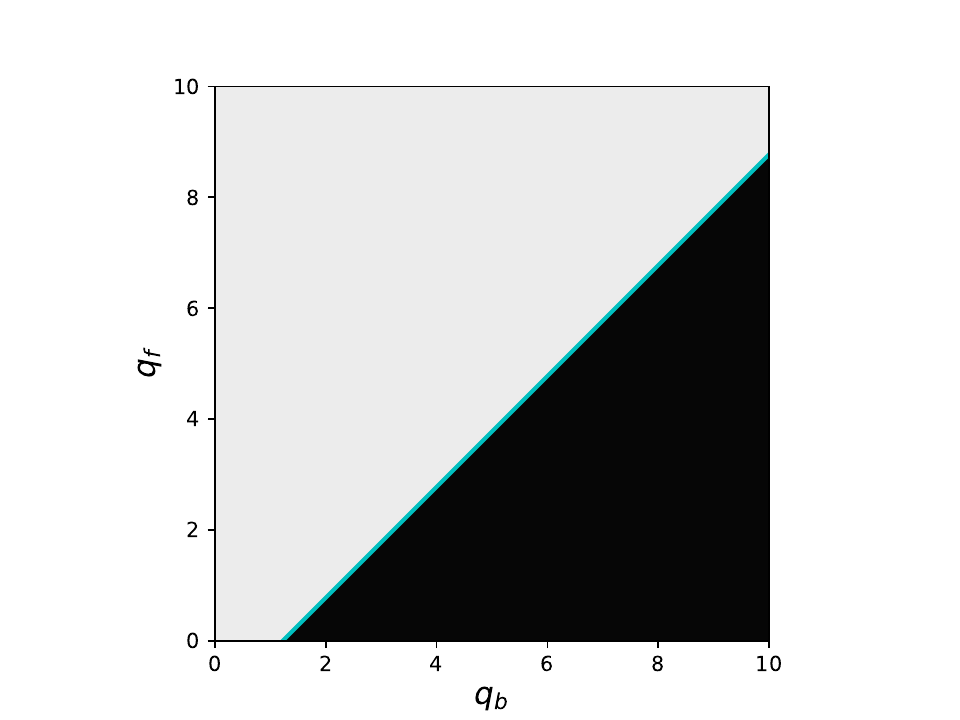}
	\caption{
	The region satisfying (black) and not satisfying (gray) the condition (\ref{app:G_assum}) with its boundary (cyan).  
	\label{conA_G}}
\end{figure} 
\begin{figure}[H]
\centering
	\includegraphics[width = 10.0cm,bb = 0 0 461 346]{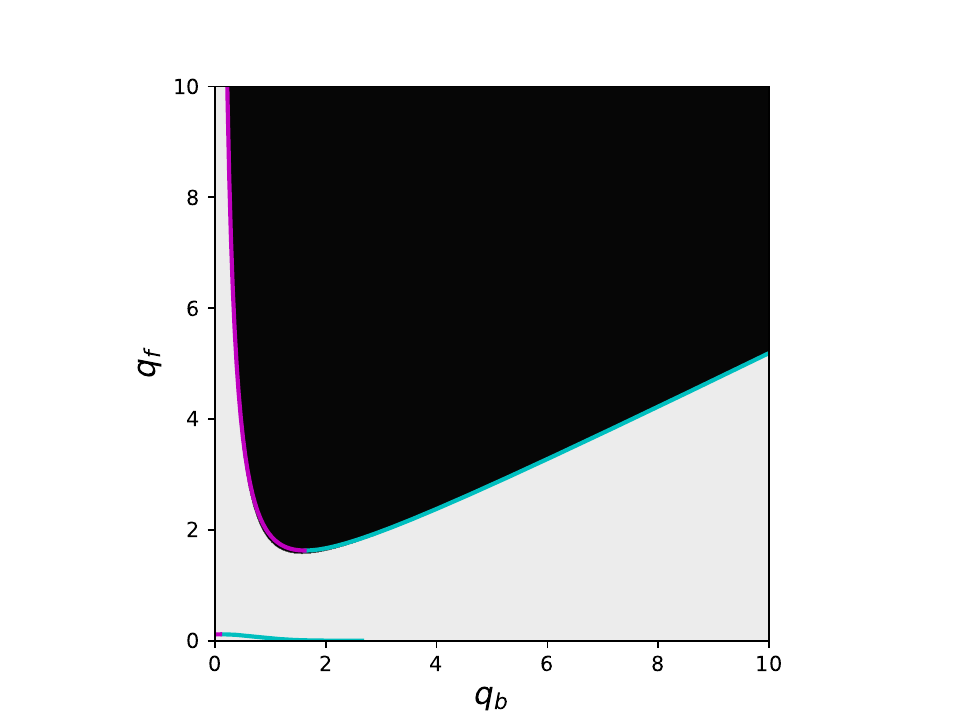}
	\caption{
		The region satisfying (black) and not satisfying (gray) the condition (\ref{eq:Appendix_conditionC})  with the boundaries (cyan and magenta) 
		for the sufficient condition (\ref{eq:Appendix_conditionE}).  
\label{conC_G}}
\end{figure} 
\begin{figure}[H]
\centering
\includegraphics[width = 10.0cm,bb = 0 0 790 592]{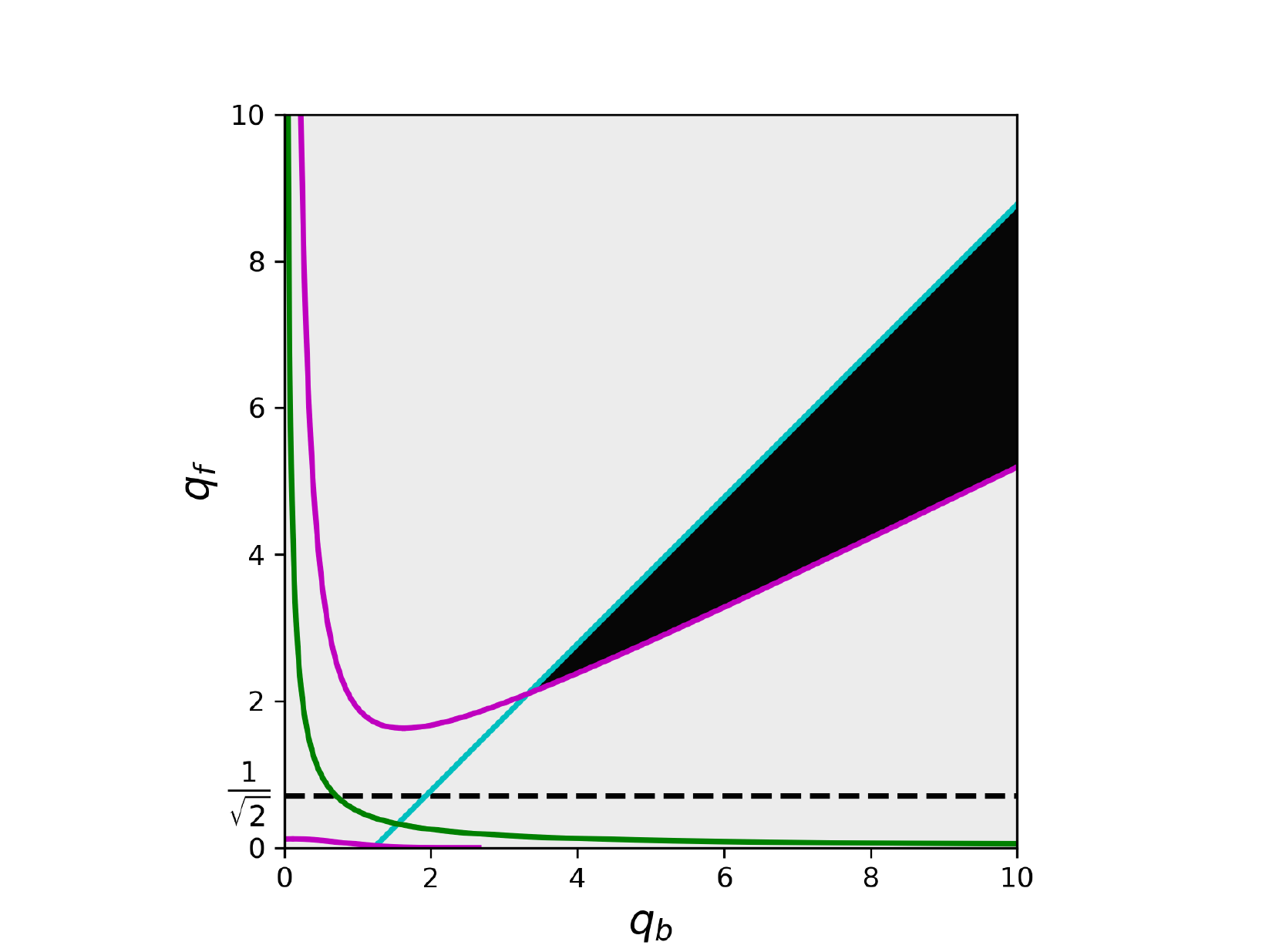}
\caption{
	The region satisfying (black) and not satisfying (gray) the conditions (\ref{eq:suf_out_g}),(\ref{eq:horseshoe_condition2}), (\ref{app:G_assum}) and (\ref{eq:Appendix_conditionC}).
	with the boundaries (colored lines) for the (sufficient) conditions.  
\label{con_all_G}}
\end{figure} 

\subsection{Redundancy of the conditions for $\kappa$}
\label{app:Gaussian_kappa}
We here seek the strongest condition among the six conditions 
given in (\ref{kappa_cond_G}) by putting a set of concrete values $(q_b,q_f)$ and $\mu$
that satisfy a set of conditions derived above.
The list of the conditions((\ref{def:kappa1}),(\ref{def:kappa2}),(\ref{def:kappa3}),(\ref{cond:kappa_4}),(\ref{def:kappa_5}) and (\ref{def:kappa_6})) is as follows:
\begin{eqnarray}
\k \geq \k_1 = \frac{4 q_f} { \displaystyle  2M_1 - \varepsilon }, \\
\k > \k_2 =\frac{ 4 q_f}{\displaystyle \frac{1}{2 M_2 }-  \varepsilon },\\
\kappa > \k_3 =\frac{\displaystyle 2 + \tilde{\mu}}{\displaystyle - 2 \left(v'\left(\left(\frac{1}{2M_2}\right)^2\right)+\frac{1}{2M_2^2}v''\left(\left(\frac{1}{2M_2}\right)^2\right)\right) -4\varepsilon M_1},\\
\kappa > \k_4 =\frac{\displaystyle 2 q_f }{\displaystyle -M_2 {q_1}^2 + q_1 -\frac{1}{2} \varepsilon},\\
 \kappa > \k_5 =\frac{\displaystyle \tilde{\mu}-2 + \frac{12aq_f}{M_2} v''(1/M_2^2)
}{\displaystyle \left( \frac{2a}{M_2^2} -2b -\frac{3a}{M_2}\varepsilon \right)v''(1/M_2^2)
},\\
\kappa > \k_6 =\frac{\displaystyle \tilde{\mu} -2 }{\displaystyle 2(aq_f^2 - b)v''(q_f^2)} .
\end{eqnarray}
As easily verified, the parameter values $(q_b,q_f)=(9/2,3)$ and $\mu = 1/2$ 
satisfy the conditions (\ref{eq:horseshoe_condition2}), (\ref{app:G_assum}), (\ref{eq:suf_out_g}) and (\ref{eq:Appendix_conditionC}) (see figure A3). 
Putting these values, we find 
\begin{eqnarray*}
\fl
(\kappa_1,\kappa_2,\kappa_3,\kappa_4 ,\kappa_5,\kappa_6) = (14.1\cdots,20.8\cdots,10.1\cdots, 21.8\cdots,25.3\cdots
,119.\cdots). 
\end{eqnarray*}
Hence the condition $\kappa > \kappa_6$ turns out to be the strongest. 
Note that this condition comes from the one associated with form of the potential 
function in the vicinity of fixed points.

\section{Conditions for the Lorentzian function}
\label
{app:Lorentzian}
\subsection{Some properties of $V'(q)$}
\label{app:Lorentzian_out}

In this Appendix, 
we show that the derivative $V'(q)$ has zeros only at $q = \pm q_f$ and $q=0$, 
and that $V'(q) < 0$ holds for $q > q_f$. 

An explicit form of the derivative $V'(q)$ is 
\begin{eqnarray}
\frac{2\kappa q}{(1+(q-q_b)^2)(1+(q-q_b)^2)} \left( \rho(q,q_b)-\rho(q_f,q_b)\right), 
\end{eqnarray}
where
\begin{eqnarray}
\rho(q,q_b) := \frac{(1+(q-q_b)^2)(1+(q-q_b)^2)}{(1+q^2)^2}. 
\end{eqnarray}
From this expression, it is clear to see that $q=\pm q_f$ and $q=0$ are zeros 
of $V'(q)$.  Here we show that these are only zeros of $V'(q)$. 
To this end, consider the derivative of $\rho(q,q_b)$, 
\begin{eqnarray}
\frac{d\rho(q,q_b)}{dq} = \frac{4q_b^2 q (q^2 - q_b^2 -3)}{(1+q^2)^3}, 
\end{eqnarray}
which takes a local maximal value $(1 + q_b^2)^2$ at  $q = 0$, 
and local minimum values $4/(q_b^2 + 4)$ at $q = \pm \sqrt{q_b^2 + 3}$, and it behaves 
asymptotically as, 
\begin{eqnarray}
\label{app:rho_asympt}
\lim_{q \rightarrow \pm \infty} \rho(q,q_b) = 1. 
\end{eqnarray}
From the expression of $\rho(q,q_b)$, we can show 
that the condition 
\begin{equation}
\label{app:rho_inequality}
1 < \rho(q_f,q_b)< (1 + q_b^2)^2
\end{equation}
is necessary and sufficient in order that $q=\pm q_f$ are only zeros of 
the function $\rho(q,q_b)-\rho(q_f,q_b)$. 
The left-hand inequality follows from the condition (\ref{eq:condout}), and 
the right inequality holds for any $q_b$ and $q_f$. 
Thus, we can say that $q=\pm q_f$ and $q=0$ are the only zeros 
of $V'(q)$. 

Since $q_f < \sqrt{q_b^2 + 3}$ trivially follows from the condition (\ref{eq:condout}),
it turns out that $\rho(q,q_b)$ monotonically decreases for $q_f < q<\sqrt{q_b^2 + 3}$ and 
monotonically increases for $\sqrt{q_b^2 + 3} < q$, and furthermore, it does not have any extrema. 
Together with (\ref{app:rho_asympt}) and (\ref{app:rho_inequality}), 
we find
\begin{eqnarray}
\rho(q,q_b) - \rho(q_f,q_b) < 0,  \quad {\rm for} \quad q_f < q, 
\end{eqnarray}
which leads to 
\begin{eqnarray}
V'(q) < 0 \quad {\rm for} \quad q_f < q. 
\end{eqnarray}
\begin{figure}[H]
\centering
\includegraphics[width = 10.0cm,bb = 0 0 461 346]{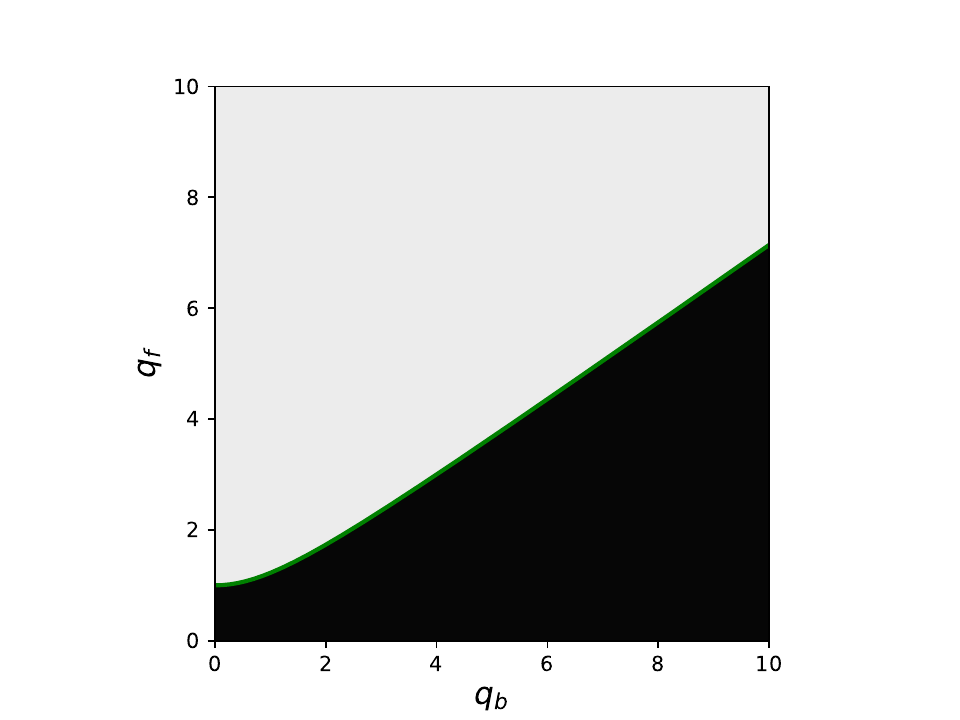}
		\caption{The region satisfying (black) and not satisfying (gray) the condition (\ref{eq:condout}) with its boundary.  
\label{conD_L}}
\end{figure} 

\subsection{The proof of Lemma \ref{lmm:extrema}}
\label{app:Lorentzian_mild}

We here prove that 
$F(q) -q_f < q_f$ or equivalently $F(q) - 2q_f < 0$ for $-q_f<q<0$, 
and $F(q)$ has a unique local maximum 
in $0<q<q_f$ for the Lorentzian potential case. 

First we rewrite $F^{(L)}(q) - q_f$ as 
\begin{eqnarray}
\fl
F^{(L)}(q)-q_f &=& -2q - q_f + \kappa \frac{2q}{(1+ q^2)^2} - \k \varepsilon \left( \frac{1}{1+(q-q_b)^2} -\frac{1}{1+(q+q_b)^2} \right) \nonumber \\
\fl
&=& -2q - q_f + \kappa \frac{2q}{(1+ q^2)^2} \widetilde{F^{(L)}}(q), 
\end{eqnarray}
where 
\begin{eqnarray}
\fl
\widetilde{F^{(L)}}(q) := 1 - \frac{\varepsilon }{2} \frac{(1+ q^2)^2}{q}    \left( \frac{1}{1+(q-q_b)^2} -\frac{1}{1+(q+q_b)^2} \right)  \nonumber 
\end{eqnarray}
Here for the Lorentzian case, 
\begin{eqnarray}
\varepsilon = \frac{(1+(q_f-q_b)^2)(1+(q_f+q_b)^2)}{2q_b(1+q_f^2)^2}. 
\end{eqnarray}
Since 
\begin{eqnarray}
\fl
\widetilde{F^{(L)}}'(q) &=& \frac{(1 + (q_f-q_b)^2)(1 + (q_f+q_b)^2)}{(1 + {q_f}^2)^2} \frac{4q_b(1 + {q}^2)q(q^2 - {q_b}^2-3)}{(1 + (q-q_b)^2)^2(1 + (q+q_b)^2)^2} , 
\end{eqnarray}
we can find the solutions for $\widetilde{F^{(L)}}'(q) = 0$ as
\begin{eqnarray}
q = 0, \pm \sqrt{{q_b}^2 + 3}. 
\end{eqnarray}
If the inequality 
\begin{eqnarray}
\label{eq:mild_L}
q_f < \sqrt{{q_b}^2 + 3}
\end{eqnarray}
is satisfied, it is easy to show that 
$\widetilde{F^{(L)}}(q)$ monotonically increases for $-q_f < q< 0$. 
Together with the fact that $\widetilde{F^{(L)}}(-q_f) = 0$, we can say that 
$\widetilde{F^{(L)}}(q)> 0$ holds for $-q_f < q < 0$. 
This implies that 
$\displaystyle \kappa \frac{2q}{(1+ q^2)^2} \widetilde{F^{(L)}}(q) < 0$ holds 
and thus $F(q)-q_f < q_f$ follows for $-q_f < q < 0$. 

Next we will show that 
${F^{(L)}}(q)$ has a unique local maximum for $0 < q < q_f$ provided that 
the condition 
\begin{eqnarray}
\label{app:L_assum}
q_f <q_b - 1
\end{eqnarray}
is satisfied (see figure B2). 
To this end we consider the behavior of the first and second derivative 
\begin{eqnarray}
\fl
\quad \quad \quad
{F^{(L)}}'(q) = -2 + 2 \kappa \frac{1 - 3q^2}{(1+q^2)^3}   +2 \k \varepsilon \left( \frac{q-q_b}{(1+(q-q_b)^2)^2} - \frac{q+q_b}{(1+(q+q_b)^2)^2}   \right) \\
\fl
\quad \quad \quad
{F^{(L)}}''(q) = \kappa \frac{24q (q^2 - 1)}{(1+q^2)^4}   +2 \k \varepsilon \left( \frac{1 - 3(q-q_b)^2}{(1+(q-q_b)^2)^3} - \frac{1 - 3(q+q_b)^2}{(1+(q+q_b)^2)^3}   \right) . 
\end{eqnarray}
We first consider the second derivative ${F^{(G)}}''(q)$. 
Due to the condition (\ref{app:L_assum}), the term
\begin{eqnarray}
\frac{1 - 3(q-q_b)^2}{(1+(q-q_b)^2)^3}
\end{eqnarray}
monotonically decreases for $0 < q < q_f$.
In a similar way, the term
\begin{eqnarray}
- \frac{1 - 3(q+q_b)^2}{(1+(q+q_b)^2)^3} 
\end{eqnarray}
monotonically decreases for $0 < q $.
Combining these with the fact that $\k>0$ and $\varepsilon>0$, 
\begin{eqnarray}
\label{app:L_kpep}
2 \k \varepsilon \left( \frac{1 - 3(q-q_b)^2}{(1+(q-q_b)^2)^3} - \frac{1 - 3(q+q_b)^2}{(1+(q+q_b)^2)^3}   \right) 
\end{eqnarray}
decreases monotonically for $0 < q < q_f$.In addition, since (\ref{app:L_kpep}) is zero at $q = 0$
\begin{eqnarray}
\label{app:L_hill}
2 \k \varepsilon \left( \frac{1 - 3(q-q_b)^2}{(1+(q-q_b)^2)^3} - \frac{1 - 3(q+q_b)^2}{(1+(q+q_b)^2)^3}   \right) <0
\end{eqnarray}
holds for $0 < q < q_f$.
As for the second derivative ${F^{(L)}}''(q)$, 
we can show from the inequality (\ref{app:L_hill}) that
\begin{eqnarray}
{F^{(L)}}''(q) < \kappa \frac{24q (q^2 - 1)}{(1+q^2)^4} .
\end{eqnarray}
Since $ \kappa 24q (q^2 - 1)/(1+q^2)^4 <0$ for $0 < q < 1$, 
the following holds for $0 < q < 1$
\begin{eqnarray}
{F^{(L)}}''(q) < 0. 
\end{eqnarray}
This implies that ${F^{(L)}}'(q)$ monotonically decreases
for $0 < q < 1$. 
From (\ref{app:L_hill}), we can show that the terms
\begin{equation}
2 \k \varepsilon \left( \frac{(q-q_b)}{(1+(q-q_b)^2)^2} - \frac{(q+q_b)}{(1+(q+q_b)^2)^2}   \right) 
\end{equation}
monotonically decreases for $0<q<q_f$, and take a negative value at $q=0$, 
Thus we have 
\begin{equation}
2 \k \varepsilon \left( \frac{(q-q_b)}{(1+(q-q_b)^2)^2} - \frac{(q+q_b)}{(1+(q+q_b)^2)^2}   \right) <0
\end{equation}
which leads to 
\begin{equation}
{F^{(L)}}'(q) < 2\kappa  \frac{1 - 3q^2}{(1+q^2)^3}
\end{equation}
for $0<q<q_f$.
Now, the fact that $ 2\kappa  (1 - 3q^2)/(1+q^2)^3< 0$ holds for $\sqrt{1/3}< q$ 
together with (\ref{app:L_hill}) leads to 
the following inequality for $\sqrt{1/3}< q$, 
\begin{eqnarray}
\label{app:L_positive}
\quad {F^{(L)}}'(q) < 0.
\end{eqnarray} 
Hence the function $F^{(L)}(q)$ has at most a single 
local maximum for $0<q<q_f$.

\subsection{Conditions for parameters $q_b$ and $q_f$}
\label
{app:Lorentzian_qf_qb}

We here examine the conditions required for $q_b$ and $q_f$ in the Lorentzian potential case. 
First we recall the condition (\ref{app:L_assum}), which is 
necessary to prove the Lemma \ref{lmm:extrema} (see \ref{app:Lorentzian_mild}).
Next we examine the inequalities required in order to derive the conditions for 
$\kappa_i  ~(1 \le i \le 5)$ (see (\ref{eq:horseshoe_condition3}),(\ref{eq:positive_discriminant}),(\ref{cond_prop62}),(\ref{eq:ep4_epsilon}) and (\ref{eq:e5_epsilon})):
\begin{eqnarray}
\varepsilon < \varepsilon^{(L)}_{1} = \frac{3\sqrt{3}}{8} \approx 0.64951,\\
\varepsilon < \varepsilon^{(L)}_{2} = \frac{27}{{25\sqrt{5}}}\approx 0.48299,\\
\varepsilon < \varepsilon^{(L)}_{3} =\frac{8}{3\sqrt{3}} \frac{5^{10}}{\displaystyle \left( 5^5 + 3^6\right)^3}\left( 5^5 - 3^7\right) \approx 0.24636
,\\
\varepsilon < \varepsilon^{(L)}_{4} =-2M_2q_1^2 + 2q_1 =-\frac{25\sqrt{5}}{81}+\frac{2}{\sqrt{3}}\approx 0.46455
,\\
\varepsilon < \varepsilon^{(L)}_{5} = \frac{2}{3M_2} - \frac{2b}{3a}M_2 = \frac{36}{25\sqrt{5}}-\frac{{25\sqrt{5}}}{243}\approx 0.41393.
\end{eqnarray}
By comparing concrete values, we obtain 
\begin{eqnarray}
\varepsilon^{(L)}_{1} > 
\varepsilon^{(L)}_{2} > 
\varepsilon^{(L)}_{4} > 
\varepsilon^{(L)}_{5} > 
\varepsilon^{(L)}_{3} , 
\end{eqnarray}
and so finally conclude that the condition 
\begin{eqnarray}
\label{eq:Appendix_conditionC_Lorentzian}
\varepsilon(q_f,q_b) < \varepsilon^{(L)}_{3}
\end{eqnarray}
turns out to be the strongest one (see figure B3).
Now using an explicit expression for $\varepsilon^{(L)}_{3}$ for the Lorentzian potential case, 
we can write down the condition (\ref{eq:Appendix_conditionC_Lorentzian}) explicitly as 
\begin{eqnarray}
\frac{1}{2q_b} \left(\frac{(1+(q_f-q_b)^2)(1+(q_f+q_b)^2)}{(1+q_f^2)^2}\right) < \varepsilon^{(L)}_3, 
\end{eqnarray}
which is recast into 
\begin{eqnarray}
\label{eq:B34}
\frac{(1+(q_f-q_b)^2)(1+(q_f+q_b)^2)}{(1+q_f^2)^2} < 2\varepsilon^{(L)}_3 q_b .
\end{eqnarray}
For the left side, we can show that 
\begin{eqnarray}
1< \frac{(1+(q_f-q_b)^2)(1+(q_f+q_b)^2)}{(1+q_f^2)^2} 
\end{eqnarray}
holds. This is obtained first by rewriting as 
\begin{eqnarray}
\frac{(1+(q_f-q_b)^2)(1+(q_f+q_b)^2)}{(1+q_f^2)^2} 
=1 + \frac{q_b^2}{(1+q_f^2)^2}(q_b^2 +2 - 2q_f^2 ), 
\end{eqnarray}
then applying the condition (\ref{eq:condout}), which leads to 
\begin{eqnarray}
q_b^2 +2 - 2q_f^2  > q_b^2 +2 - 2\left( \frac{q_b^2}{2} + 1 \right) = 0. 
\end{eqnarray}
Therefore, in order to consider the condition (\ref{eq:Appendix_conditionC_Lorentzian}), 
the following should be satisfied.
\begin{eqnarray}
\label{eq:a38}
1 < 2\varepsilon^{(L)}_3 q_b. 
\end{eqnarray}

Now note that the condition (\ref{eq:B34}) is rewritten as
\begin{eqnarray}
\fl
~~~~~~~~
(2\varepsilon^{(L)}_3 q_b-1)q_f^4 +2(q_b^2 +2\varepsilon^{(L)}_3 q_b -1)q_f^2 - (q_b^4+q_b^2-2\varepsilon^{(L)}_3 q_b + 1)> 0 .
\end{eqnarray}
Considering that the condition (\ref{eq:a38}) and that $q_f>0$ holds, 
we can explicitly solve the above inequality for $q_f$ as 
\begin{eqnarray}
q_f >  \sqrt{ \frac{\displaystyle -q_b^2 +q_b\sqrt{2\varepsilon^{(L)}_3 q_b^3 + 8\varepsilon^{(L)}_3 q_b -4}}{2\varepsilon^{(L)}_3 q_b-1} - 1}. 
\end{eqnarray}

Recalling the already imposed conditions 
(\ref{eq:condout}), and (\ref{con_mild}), we summarize the conditions as, 
\begin{eqnarray}
&&q_f >  \sqrt{ \frac{\displaystyle -q_b^2 +q_b\sqrt{2\varepsilon^{(L)}_3 q_b^3 + 8\varepsilon^{(L)}_3 q_b -4}}{2\varepsilon^{(L)}_3 q_b-1} - 1}, \\
&&q_f < \sqrt{\frac{1}{2}q_b^2 + 1}, \\
&&q_f < q_b-1.
\end{eqnarray}

We finally check whether there actually exists 
the domain $(q_b,q_f)$ in which the above three inequalities are satisfied. 
To this end, we rewrite the right side of the first one as 
\begin{eqnarray}
\fl
 \sqrt{ \frac{\displaystyle -q_b^2 +q_b\sqrt{2\varepsilon^{(L)}_3 q_b^3 + 8\varepsilon^{(L)}_3 q_b -4}}{2\varepsilon^{(L)}_3 q_b-1} - 1}
= \sqrt{\frac{q_b^3 + 4q_b}{\displaystyle \sqrt{2\varepsilon^{(L)}_3 q_b^3 + 8\varepsilon^{(L)}_3 q_b -4} + q_b}- 1}, 
\end{eqnarray}
and compare this with the right-hand side of the second inequality. 
After some calculations, we reach the condition 
\begin{eqnarray}
(2\varepsilon^{(L)}_3 q_b-1)(q_b^2 + 4) > 0, 
\end{eqnarray}
in order to find a non-empty domain $(q_b,q_f)$ 
ensuring that the first and second inequalities both hold. 
However it is easy see that 
this is always satisfied as $q_b > 1/2\varepsilon^{(L)}_3 \approx 2.029$. 
As for the second and third inequalities, 
it is also easy to see that 
\begin{eqnarray}
\sqrt{\frac{1}{2}q_b^2 + 1} < q_b-1
\end{eqnarray}
holds for $q_b > 4$. 
Therefore, the conditions  we have to impose are summarized as 
\begin{eqnarray}
4< q_b , 
\end{eqnarray}
and 
\begin{eqnarray}
\sqrt{ \frac{\displaystyle -q_b^2 +q_b\sqrt{2\varepsilon^{(L)}_3 q_b^3 + 8\varepsilon^{(L)}_3 q_b -4}}{2\varepsilon^{(L)}_3 q_b-1} - 1}< q_f <\sqrt{\frac{1}{2}q_b^2 + 1} .
\end{eqnarray}

\begin{figure}[H]

  \centering
\includegraphics[width = 10.0cm,bb = 0 0 461 346]{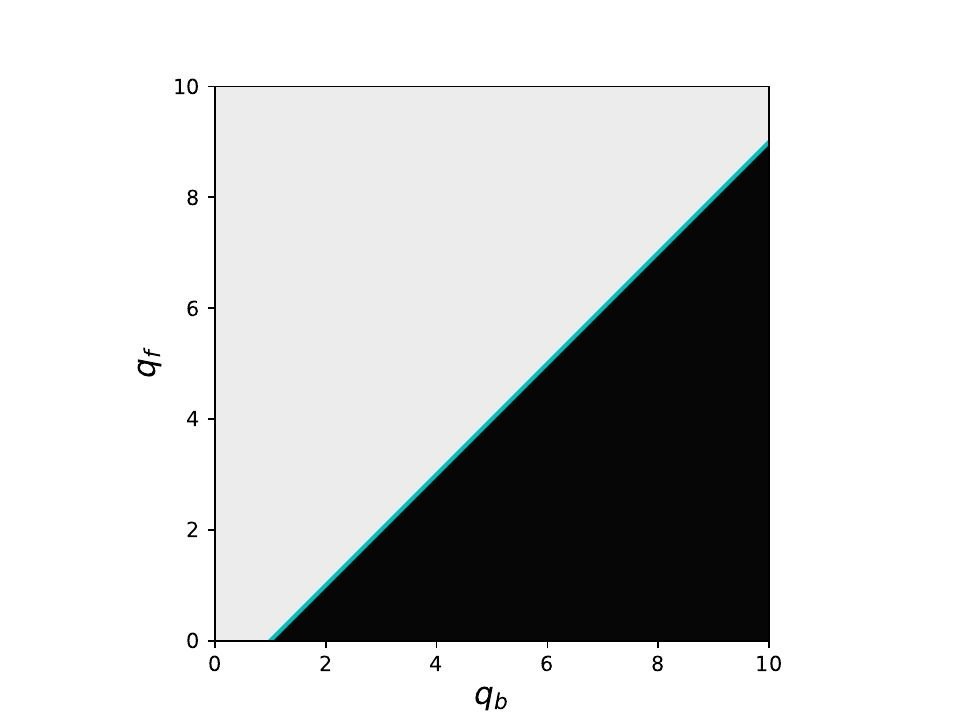}
	\caption{The region satisfying (black) and not satisfying (gray) the condition  (\ref{app:L_assum}) with its boundary (cyan). \label{cond_A_L}}
\end{figure} 
\begin{figure}[H]

  \centering
\includegraphics[width = 10.0cm,bb = 0 0 461 346]{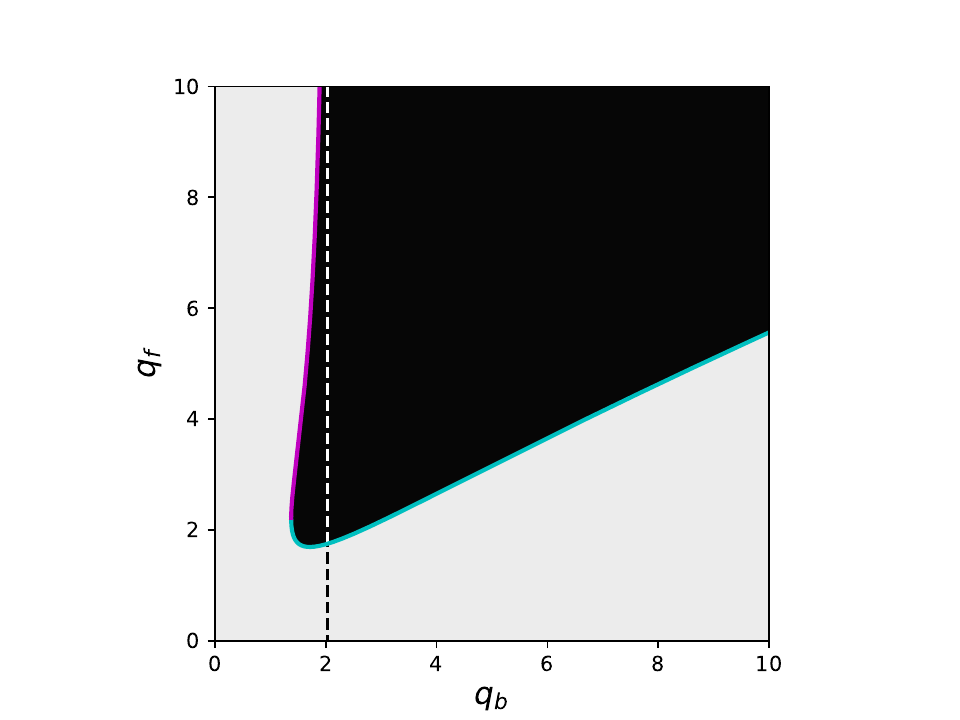}
	\caption{The region satisfying (black) and not satisfying (gray) the condition (\ref{eq:Appendix_conditionC_Lorentzian}) with the lower boundary for the condition (\ref{eq:Appendix_conditionC_Lorentzian}) (cyan). The dashed line shows $q_b = 1/2\varepsilon^{(L)}_3$}. \label{cond_B_L}
\end{figure} 
\begin{figure}[H]
        \centering
\includegraphics[width = 10.0cm,bb = 0 0 790 592]{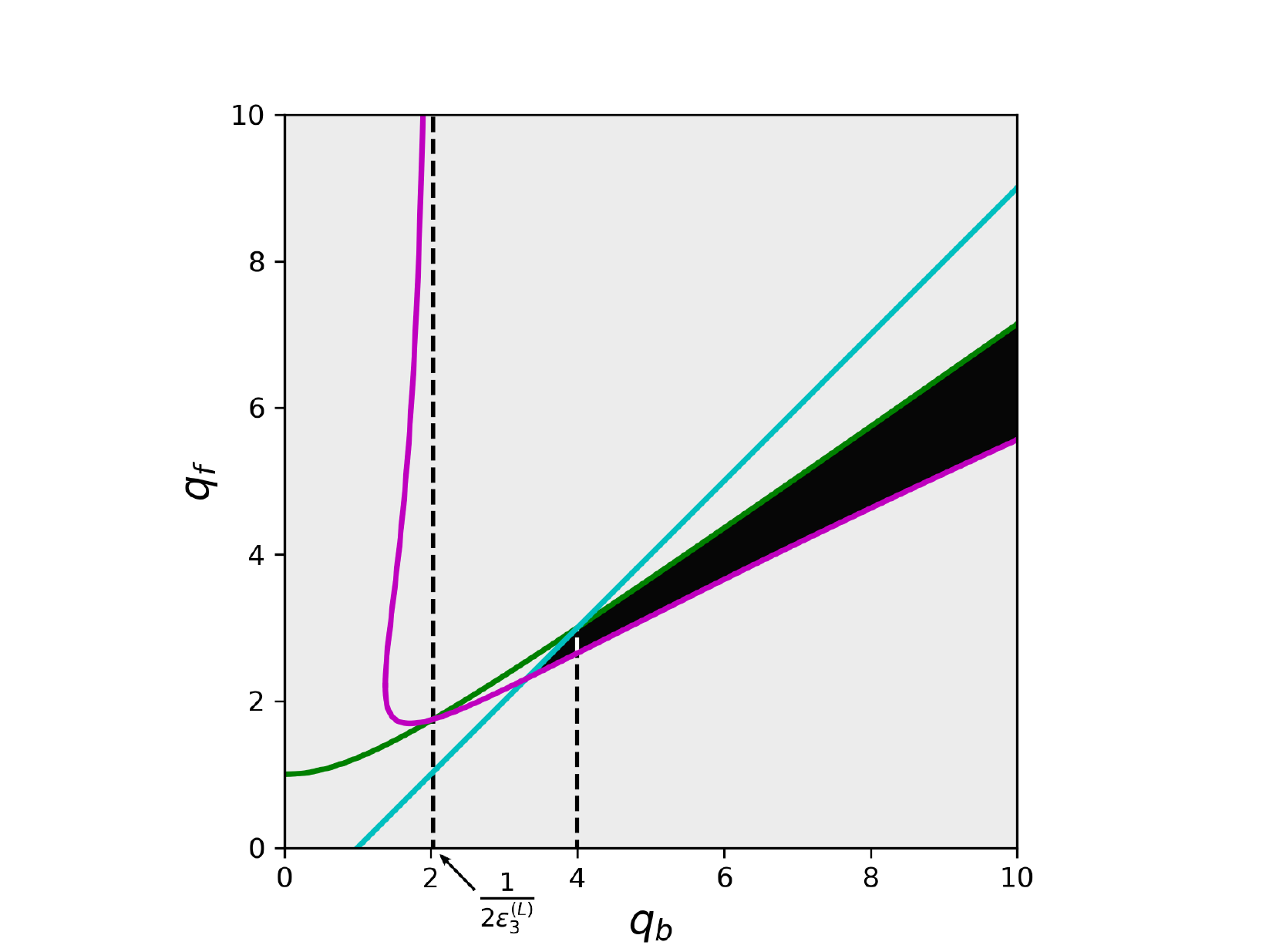}
	\caption{The region satisfying (black) and not satisfying (gray) the conditions (\ref{eq:condout}),(\ref{eq:horseshoe_condition2}) (\ref{app:L_assum}) and (\ref{eq:Appendix_conditionC_Lorentzian}) with boundaries for the  (sufficient) conditions (colored lines). \label{cond_all_L}}
\end{figure} 
\subsection{Redundancy of the conditions for $\k$}
\label{app:Lorentzian_kappa}
We seek the strongest condition among the six conditions  
given in (\ref{kappa_cond_L}) by putting a set of concrete values $(q_b,q_f)$ and $\mu$
that satisfy a set of conditions derived above.
The list of the conditions is as follows:
\begin{eqnarray}
\k \geq \kappa^{(L)}_1 = \frac{4 q_f} { \displaystyle  2M_1 - \varepsilon },  \\
\k >\kappa^{(L)}_2 = \frac{ 4 q_f}{\displaystyle \frac{1}{2 M_2 }-  \varepsilon }, \\
\kappa >\kappa^{(L)}_3 =  \frac{\displaystyle 2 + \tilde{\mu}}{\displaystyle - 2 \left(v'\left(\left(\frac{1}{2M_2}\right)^2\right)+\frac{1}{2M_2^2}v''\left(\left(\frac{1}{2M_2}\right)^2\right)\right) -4\varepsilon M_1},\\
\kappa >\kappa^{(L)}_4 =  \frac{\displaystyle 2 q_f }{\displaystyle -M_2 {q_1}^2 + q_1 -\frac{1}{2} \varepsilon},\\
 \kappa > \kappa^{(L)}_5 = \frac{\displaystyle \tilde{\mu}-2 + \frac{12aq_f}{M_2}v''(1/M_2^2)
}{\displaystyle \left( \frac{2a}{M_2^2} -2b -\frac{3a}{M_2}\varepsilon \right)v''(1/M_2^2)
},\\
\kappa > \kappa^{(L)}_6 = \frac{\displaystyle \tilde{\mu} -2 }{\displaystyle 2(aq_f^2 - b)v''(q_f^2)}.
\end{eqnarray}
As easily verified, the parameter values $(q_b,q_f)=(9/2,3)$ and $\mu = 1/2$ 
satisfy the conditions (\ref{eq:horseshoe_condition2}) (\ref{app:L_assum}),(\ref{eq:condout}) and (\ref{eq:Appendix_conditionC_Lorentzian}) (see figure B4). 
Putting these values, we have 
\begin{eqnarray*}
\fl
(\kappa_1,\kappa_2,\kappa_3,\kappa_4 ,\kappa_5,\kappa_6) = (27.1\cdots,43.4\cdots,87.4\cdots, 46.5\cdots,58.4\cdots 
,9.61\cdots). 
\end{eqnarray*}
Hence the condition $\kappa > \kappa_3$ turns out to be the strongest. 
Note that, not like the Gaussian case, this condition does not necessarily 
come from the one associated with form of the potential 
function in the vicinity of fixed points. 


\end{document}